\documentclass[preprint,authoryear]{elsarticle}

\usepackage[margin=1in]{geometry}
\usepackage[T1]{fontenc}
\usepackage[utf8]{inputenc}
\usepackage{lmodern}
\usepackage{microtype}
\usepackage{amsmath,amssymb,amsthm,mathtools,bm}
\usepackage{booktabs,tabularx,multirow,longtable}
\usepackage{graphicx}
\usepackage[table]{xcolor}
\usepackage{float}
\usepackage{algorithm}
\usepackage{algpseudocode}

\usepackage{hyperref}

\graphicspath{{note_figures/}}

\newcommand{\R}{\mathbb{R}}
\newcommand{\C}{\mathbb{C}}

\newcommand{\Cov}{\operatorname{Cov}}
\newcommand{\diag}{\operatorname{diag}}
\newcommand{\CN}{\mathcal{CN}}
\newcommand{\Normal}{\mathcal{N}}
\newcommand{\Matern}{\mathcal{M}}
\newcommand{\ii}{\mathrm{i}}
\newcommand{\trans}{\mathsf{T}}
\newcommand{\herm}{\mathsf{H}}

\DeclareMathOperator*{\argmin}{arg\,min}

\newtheorem{proposition}{Proposition}[section]
\theoremstyle{remark}

\begin{document}

\begin{frontmatter}

\title{Temporal Fourier Likelihoods with Spatial Hilbert-Space Gaussian Process Approximations}

\author[umu]{Xin Huang\corref{cor1}}
\ead{xin.huang@umu.se}

\author[psu]{Jia Li}
\ead{jol2@psu.edu}

\author[umu]{Jun Yu}
\ead{jun.yu@umu.se}

\cortext[cor1]{Corresponding author}

\address[umu]{Department of Mathematics and Mathematical Statistics, Ume{\aa} University, 901 87 Ume{\aa}, Sweden}
\address[psu]{Department of Statistics, The Pennsylvania State University, University Park, PA 16802, USA}

\begin{abstract}
Temporal Fourier methods simplify inference for stationary space-time Gaussian processes observed at many sites on a common regular time grid, but they leave a spatial covariance calculation for each temporal Fourier mode.  When a model is specified by a joint spectral density whose half-spectrum has no convenient closed form, constructing these covariances would require repeated spatial Fourier inversion.  We combine a temporal discrete Fourier transform (DFT) and Whittle likelihood with a spatial Hilbert-space Gaussian process (HSGP) approximation, which forms the spatial covariance matrices by evaluating the radial profile of the sampled spectral multiplier at Laplacian eigenfrequencies.  To account for variation in the required spatial resolution across temporal frequencies, we develop a frequency-adaptive framework that selects the HSGP rank separately for each mode.  Two complementary error criteria guide this rank selection, with subsequent reassessment at updated covariance parameter estimates. Benchmark experiments compare the proposed approximations with analytic full-rank and numerical inversion references. For a benchmark model specified directly through its spectral multiplier and requiring numerical spatial Fourier inversion, the adaptive method reduced mean fitting and reconstruction time by nearly half while maintaining comparable accuracy in reconstructing the latent field. Simulations and wind-field applications further demonstrate that sufficiently rich bases can provide accurate held-site reconstruction, and illustrate the settings in which frequency adaptation can reduce repeated spatial computation.
\end{abstract}

\begin{keyword}
Spatio-temporal Gaussian processes \sep temporal discrete Fourier transform \sep half-spectral methods \sep Hilbert-space Gaussian process approximation \sep adaptive rank selection
\end{keyword}

\end{frontmatter}

\section{Introduction}

Environmental and geophysical variables are often observed at fixed spatial locations on a common regular time grid.  Monitoring networks, satellite products, and numerical reanalyses therefore produce collections of spatially indexed time series rather than isolated spatial or temporal observations \citep{stein2005regular,chengentonsun2021}.  A space-time Gaussian process model uses the covariance structure to estimate dependence, reconstruct the field at unobserved locations, and quantify predictive uncertainty.  Direct likelihood calculations become expensive as either the number of sites or the record length increases.  The regular temporal design, however, offers structure that can reduce this cost without discarding spatial dependence.

Space-time covariance models differ mainly in how they represent interaction between the spatial and temporal domains.  Separable covariances are convenient but require the spatial correlation structure to remain unchanged over temporal frequency.  Nonseparable models may instead be constructed in the covariance domain or through a joint spectrum.  \citet{cressiehuang1999} used spectral arguments and Fourier inversion to construct stationary nonseparable covariances, while \citet{gneiting2002} developed broad covariance-domain classes.  \citet{stein2005} examined regularity and interaction properties of spectral constructions.  \citet{fuentes2008} specified nonseparable dependence through a joint spectrum and extended the formulation to spatially nonstationary settings.  These works provide concrete and fundamental constructions of valid space-time covariance models.  A spectral formulation of the covariance is especially natural when scientific assumptions concern how variation is distributed over spatial and temporal scales.

Regular temporal sampling makes the half-spectral representation particularly useful.  In this representation, the temporal lag is transformed while the spatial lag is retained, so the half-spectrum at each temporal frequency is a spatial covariance kernel \citep{stein2005regular,horrellstein2017}.  Applying a discrete Fourier transform (DFT) along the temporal dimension produces a vector of spatial coefficients at each Fourier mode.  A Whittle likelihood then treats the coefficient vectors at distinct modes as independent spatial problems \citep{whittle1953}.  This approximation decouples the temporal calculations, but a spatial covariance matrix must still be formed and factorized for every mode and every trial parameter value.  When a model is specified by a joint spectral density, the corresponding spatial covariance is obtained by an inverse Fourier transform over spatial frequency.  If this transform has no convenient closed form, the half-spectrum must be evaluated repeatedly by numerical Fourier or Hankel quadrature over many spatial lags, temporal frequencies, and parameter values.

Several methods have developed inference in this intermediate domain.  \citet{subbarao2014} formed an approximate likelihood from the joint complex distribution of finite Fourier transforms at multiple sites.  \citet{mosammamkent2016} developed estimation and diagnostic procedures for covariance-spectral models.  Related composite-likelihood formulations avoided inversions of large covariance matrices by using differenced, paired, and conditional frequency-domain contributions \citep{mosammam2016,mosammam2025}. \citet{geoga2021} combined half-spectral covariance modeling with efficient time-domain likelihood calculations for nonstationary high-frequency data. Together, these studies show that temporal spectral representations support flexible modeling and scalable inference. Their computations generally assume that the required half-spectrum, coherence function, or pairwise spectral quantity can be evaluated directly. This study considers a complementary setting in which the joint spectral multiplier is readily available but its spatial inverse transform is inconvenient, and seeks an efficient construction of the resulting spatial covariance at each temporal frequency.

% The Hilbert-space Gaussian process (HSGP) approximation provides a natural way to construct these matrices directly from their spectral multipliers \citep{solinsarkka2020}.
We construct the frequency-specific spatial covariances using a Hilbert-space Gaussian process (HSGP) approximation \citep{solinsarkka2020}.
%For a radial spatial multiplier, HSGP evaluates its radial profile at Laplacian eigenfrequencies on an extended bounded domain and uses the resulting values to weight a finite set of Laplacian eigenfunctions.  
For a stationary spatial covariance with radial spectral density, HSGP evaluates its radial profile at the square roots of the retained eigenvalues of a Laplacian operator on an extended bounded domain. These values weight the corresponding Laplacian eigenfunctions to form a finite-dimensional covariance approximation. Because the eigenfunctions depend on the domain geometry rather than the covariance parameters, their evaluations and the associated feature projections can be reused during likelihood optimization.  Approximation accuracy depends on both the domain extension and the number of retained basis functions \citep{riutortmayol2023}.  Related work has used Kullback--Leibler bounds to study rank truncation in Gaussian process approximations \citep{daskalakis2022} and has considered dimension reduction in other operator-based spatial expansions \citep{caviereskrumscheid2026}.

In a half-spectral likelihood, the ordered spatial basis is shared across temporal frequencies, while its spectral weights vary with temporal frequency and with the covariance parameters.  The rank needed for an accurate approximation can therefore differ across Fourier modes.  A common fixed rank must accommodate the most demanding mode and may retain unnecessary basis functions elsewhere.  This motivates a frequency-adaptive HSGP approximation that selects a separate rank from one nested basis at each retained mode.  The transformation to the half-spectral domain keeps exact, whereas approximation enters through Whittle decoupling in time and finite-basis truncation in space.

We develop two criteria for choosing the frequency-specific ranks.  The first controls normalized Kullback--Leibler divergence relative to a stable higher-rank HSGP covariance. The second replaces direct divergence evaluation with a covariance-scaled trace bound, which gives a sufficient criterion for the same comparison.  
Both approaches begin with an initial covariance-parameter
estimate obtained from a lower-rank pilot fit.
Each then applies its respective criterion to select
frequency-specific ranks and re-estimates the covariance
parameters using the resulting rank profile.
The ranks are subsequently reassessed at the updated
parameter estimates. If any rank fails to meet the
prescribed criterion, it is increased and the model is
refitted. This iterative process continues until no
further rank increases are required.
The final rank profile thus adapts to both the temporal frequency and to the covariance structure inferred from the observed data. Final likelihood evaluations and spatial reconstruction use these frequency-specific ranks, which can reduce computation when they remain appreciably below the reference rank used in the stable HSGP fit.  Analytic full-rank and numerical-inversion covariances provide external benchmarks in our selected experiments.

The main contributions of this work are summarized as follows.
\begin{enumerate}
\item \textit{Frequency-adaptive spatial approximation.}
We introduce two model-based criteria for selecting frequency-specific ranks from a common nested Laplacian basis.  The ranks are reassessed after covariance parameter estimation and enlarged when necessary, allowing the spatial approximation to respond to changes in the fitted multiplier across temporal frequencies.

\item \textit{A continuous-to-sampled formulation.}
We give a consistent derivation connecting the continuous joint spectrum, temporal sampling and aliasing, the sampled half-spectrum, and the covariance of finite-record DFT coefficients.  The formulation identifies the spectral weights used by HSGP and distinguishes finite-record covariance from the circular model and the Whittle working likelihood.

\item \textit{Likelihood, reconstruction, and empirical assessment.}
We develop likelihood and held-site reconstruction calculations that reuse projections onto the shared spatial basis.  Simulations assess approximation accuracy against analytic full-rank and numerical inversion references.  Wind-field applications and spatial interpolation baselines evaluate predictive accuracy and observed computation.  Together, the experiments distinguish approximation, estimation, and reconstruction effects and show how the selected ranks influence both accuracy and computational cost.
\end{enumerate}

The remainder of the paper is organized as follows.  Section~\ref{sec:half-spectral} introduces the continuous temporal half-spectrum and its relation to the joint space-time spectrum.  Section~\ref{sec:hsgp} develops temporal sampling, the Whittle working model, the spatial HSGP approximation-based reconstruction, frequency-adaptive rank selection, and computational cost.  Section~\ref{sec:numerics} presents the benchmark models, simulation studies, and wind-field applications.  The final section discusses the findings and the practical scope of the method.  The appendices collect Fourier conventions, finite-record covariance results, supporting operator arguments, inference calculations, and numerical specifications.

\section{Continuous half-spectral representation}
\label{sec:half-spectral}

\subsection{Observation model and continuous-time half-spectrum}

Let $Z(\bm s,t)$ be a real-valued, zero-mean stationary Gaussian random field on $\R^d\times\R$.\footnote{ The zero-mean assumption simplifies the presentation of the covariance construction. For a nonzero mean, the same covariance calculations apply to residuals under a specified mean model.  \ref{app:inference} describes the estimation of mean coefficients and the associated restricted likelihood. } Its covariance function is defined as
\begin{equation}
 C_Z(\bm h,u;\bm\theta)
 :=\Cov\big(Z(\bm s,t),Z(\bm s-\bm h,t-u)\big),
 \label{eq:latent-covariance}
\end{equation}
where $d\geq 1$ denotes the spatial dimension, $\bm h\in\R^d$ and $u\in\R$ are respectively the spatial and temporal lags, and $\bm\theta\in\Theta$ is the covariance parameter vector, with $\Theta$ denoting the admissible parameter space. The observations are collected at spatial locations $\bm s_i$, $i=1,\ldots,n$, and at regularly spaced time points
\[
 t_j=j\Delta t,\qquad j=0,\ldots,T-1,\qquad \Delta t>0.
\]
The observation model is given by
\begin{equation}
 Y(\bm s_i,t_j)
 =Z(\bm s_i,t_j)+\mathcal E(\bm s_i,t_j),
 \qquad
 \mathcal E(\bm s_i,t_j)\sim\Normal(0,v_{\mathrm n}),
 \label{eq:observation-model}
\end{equation}
where the measurement errors $\mathcal E(\bm s_i,t_j)$ are independent and identically distributed across spatial locations and time points, with variance $v_{\mathrm n}>0$, and are independent of the latent Gaussian field $Z$. When estimated, the measurement-error variance $v_{\mathrm n}$ is also included in $\bm\theta$.

The temporal half-spectrum associated with $C_Z$ is obtained by transforming the covariance function with respect to temporal lag while retaining spatial lag. For each spatial lag, assuming that $u\mapsto C_Z(\bm h,u;\bm\theta)$ is integrable and using the symmetric unitary Fourier convention, define
\begin{equation}
 H_{\mathrm c}(\bm h,\tau;\bm\theta)
 =\frac{1}{\sqrt{2\pi}}\int_{\R}
 C_Z(\bm h,u;\bm\theta)e^{-\ii\tau u}\,\mathrm du,
 \label{eq:half-spectrum-definition}
\end{equation}
where $\tau\in\R$ denotes the temporal angular frequency. Under the conditions for Fourier inversion, the covariance function is recovered as
\begin{equation}
 C_Z(\bm h,u;\bm\theta)
 =\frac{1}{\sqrt{2\pi}}\int_{\R}
 H_{\mathrm c}(\bm h,\tau;\bm\theta)e^{\ii\tau u}\,\mathrm d\tau.
 \label{eq:half-spectrum-inversion}
\end{equation}
We refer to $H_{\mathrm c}$ as the \emph{continuous-time temporal half-spectrum}, where the subscript $\mathrm c$ denotes continuous time. The term ``half-spectrum'' indicates that only the temporal lag is transformed into the frequency domain, while the spatial lag $\bm h$ is retained.  All temporal frequencies $\tau\in\R$ are kept.

\subsection{From the joint spectrum to the half-spectrum}
\label{sec:spatial-spectrum}

Suppose that the stationary covariance function admits an integrable, nonnegative joint spectral density $f(\bm\omega,\tau;\bm\theta)$, where $\bm\omega\in\R^d$ denotes the spatial angular-frequency vector. Under the same symmetric unitary Fourier convention, the covariance function admits the representation
\begin{equation}
 C_Z(\bm h,u;\bm\theta)
 =\frac{1}{(2\pi)^{(d+1)/2}}
 \int_{\R}\int_{\R^d}
 e^{\ii(\bm\omega^\trans\bm h+\tau u)}
 f(\bm\omega,\tau;\bm\theta)
 \,\mathrm d\bm\omega\,\mathrm d\tau.
 \label{eq:joint-fourier-representation}
\end{equation}
We also assume that $\bm\omega\mapsto f(\bm\omega,\tau;\bm\theta)$ is integrable for each temporal frequency $\tau$ considered below. Since $Z$ is real-valued, its spectral density satisfies $f(-\bm\omega,-\tau;\bm\theta) =f(\bm\omega,\tau;\bm\theta)$. Setting $\bm h=\bm0$ and $u=0$ in \eqref{eq:joint-fourier-representation} shows that the integral of $f$ determines the latent variance $v_{\mathrm f}=C_Z(\bm0,0;\bm\theta)$.  Hence the assumed integrability of $f$ ensures a finite latent variance \citep[see also][]{cressiehuang1999}.

Taking the temporal Fourier transform of \eqref{eq:joint-fourier-representation} yields the inverse spatial Fourier representation of the half-spectrum:
\begin{equation}
 H_{\mathrm c}(\bm h,\tau;\bm\theta)
 =\frac{1}{(2\pi)^{d/2}}\int_{\R^d}
 e^{\ii\bm\omega^\trans\bm h}
 f(\bm\omega,\tau;\bm\theta)\,\mathrm d\bm\omega.
 \label{eq:joint-to-continuous-half}
\end{equation}
Conversely, under suitable conditions for spatial Fourier inversion, the joint spectral density can be recovered as
\begin{equation}
 f(\bm\omega,\tau;\bm\theta)
 =\frac{1}{(2\pi)^{d/2}}\int_{\R^d}
 e^{-\ii\bm\omega^\trans\bm h}
 H_{\mathrm c}(\bm h,\tau;\bm\theta)\,\mathrm d\bm h
 =:\widehat H_{\mathrm c}(\bm\omega,\tau;\bm\theta).
 \label{eq:half-spatial-transform}
\end{equation}
For each fixed temporal frequency $\tau$, Equation~\eqref{eq:joint-to-continuous-half} is the inverse spatial Fourier transform of the nonnegative function $f(\cdot,\tau;\bm\theta)$.  Bochner's theorem therefore implies that $H_{\mathrm c}(\cdot,\tau;\bm\theta)$ is a Hermitian positive-semidefinite kernel \citep{wendland2005}: for any locations $\bm s_1,\ldots,\bm s_n$ and complex numbers $a_1,\ldots,a_n$,
\[
 \sum_{r=1}^n\sum_{\ell=1}^n
 a_r\,\overline{a_\ell}\,
 H_{\mathrm c}(\bm s_r-\bm s_\ell,\tau;\bm\theta)\geq0.
\]
The half-spectrum therefore can be viewed as a family of possibly complex-valued spatial covariance kernels indexed by temporal frequencies.

\section{Temporal DFT and spatial Hilbert-space approximation}
\label{sec:hsgp}

\subsection{Continuous-frequency spatial covariance operator}

HSGP weights its basis functions by the spectral multiplier, or symbol, of the spatial covariance operator.  The spatial Fourier pair in Section~\ref{sec:spatial-spectrum} identifies this symbol, which is recorded in the subsequent proposition as a standard consequence of the convolution theorem, and Section~\ref{sec:hsgp-covariance} specializes it to radial multipliers.

\begin{proposition}[Continuous-frequency spatial symbol]
\label{prop:continuous-spatial-symbol}
Fix $\tau\in\R$ and $\bm\theta\in\Theta$.  Suppose that the continuous-time temporal half-spectrum $H_{\mathrm c}(\cdot,\tau;\bm\theta)\in L^1(\R^d)$ and that its unitary spatial Fourier transform is $f(\cdot,\tau;\bm\theta)$.  Define the spatial convolution operator
\begin{equation}
 [\mathcal K^{\mathrm c}_{\tau,\bm\theta}\,g](\bm s)
 :=\int_{\R^d}
 H_{\mathrm c}(\bm s-\bm s',\tau;\bm\theta)
 \,g(\bm s')\,\mathrm d\bm s',
 \qquad g\in L^2(\R^d;\C).
 \label{eq:continuous-covariance-operator}
\end{equation}
Then $\mathcal K^{\mathrm c}_{\tau,\bm\theta}$ is bounded on $L^2(\R^d;\C)$ and acts in the spatial Fourier domain as
\begin{equation}
 \widehat{\mathcal K^{\mathrm c}_{\tau,\bm\theta}g}(\bm\omega)
 =S_{\mathrm c}(\bm\omega,\tau;\bm\theta)\,\widehat g(\bm\omega),
 \label{eq:continuous-multiplier-action}
\end{equation}
where $S_{\mathrm c}$, the spectral multiplier (or symbol) of the operator, is
\begin{equation}
 S_{\mathrm c}(\bm\omega,\tau;\bm\theta)
 =(2\pi)^{d/2}\,f(\bm\omega,\tau;\bm\theta).
 \label{eq:continuous-symbol-density-relation}
\end{equation}
If, in addition, $f(\cdot,\tau;\bm\theta)\in L^1(\R^d)$, then Fourier inversion gives, at every continuity point of $H_{\mathrm c}$,
\begin{equation}
 H_{\mathrm c}(\bm h,\tau;\bm\theta)
 =\frac{1}{(2\pi)^d}\int_{\R^d}
 e^{\ii\bm\omega^\trans\bm h}
 \,S_{\mathrm c}(\bm\omega,\tau;\bm\theta)\,\mathrm d\bm\omega.
 \label{eq:continuous-symbol-inversion}
\end{equation}
\end{proposition}

\begin{proof}
Fix $\tau$ and $\bm\theta$ and suppress these arguments temporarily.  The operator in \eqref{eq:continuous-covariance-operator} is the spatial convolution $H_{\mathrm c}*g$, where $*$ denotes convolution on $\R^d$. Young's inequality \citep{grafakos2014} gives
\[
 \|\mathcal K^{\mathrm c}_{\tau,\bm\theta}g\|_{L^2}
 =\|H_{\mathrm c}*g\|_{L^2}
 \leq \|H_{\mathrm c}\|_{L^1}\|g\|_{L^2},
\]
which proves boundedness. For $g\in L^1(\R^d)\cap L^2(\R^d)$, the convolution theorem and \eqref{eq:half-spatial-transform} yield
\begin{align*}
 \widehat{\mathcal K^{\mathrm c}_{\tau,\bm\theta}g}
 &=\widehat{H_{\mathrm c}*g}
 =(2\pi)^{d/2}\widehat H_{\mathrm c}\,\widehat g\\
 &=(2\pi)^{d/2}f\,\widehat g
 =S_{\mathrm c}\,\widehat g.
\end{align*}
Since $H_{\mathrm c}\in L^1$, its Fourier transform is bounded.  The identity therefore extends from the dense subspace $L^1\cap L^2$ to every $g\in L^2$.  If $f\in L^1$ as well, the Fourier inversion theorem and \eqref{eq:continuous-symbol-density-relation} give \eqref{eq:continuous-symbol-inversion}.
\end{proof}

Equations~\eqref{eq:continuous-symbol-density-relation} and \eqref{eq:continuous-symbol-inversion} show that $f$ and $S_{\mathrm c}$ encode the same spatial-frequency model and differ only in normalization.  The function $f$ is the joint spectral density under the unitary Fourier convention, whereas $S_{\mathrm c}$ is the factor that multiplies $\widehat g$ in the covariance-operator action \eqref{eq:continuous-multiplier-action}.  On $\R^d$, plane waves $e^{\ii\bm\omega^\trans\bm s}$ are generalized eigenfunctions of both $\mathcal K^{\mathrm c}_{\tau,\bm\theta}$ and the negative spatial Laplacian operator, with generalized eigenvalues $S_{\mathrm c}(\bm\omega,\tau;\bm\theta)$ and $\|\bm\omega\|^2$, respectively.  Under a radial multiplier, the covariance operator is therefore a scalar function of the negative Laplacian.  This whole-space relation motivates the bounded-domain construction in Section~\ref{sec:hsgp-covariance}.

The operator construction above is continuous in time, whereas the data are observed on a finite regular grid.  The next subsection connects these two settings.  It first derives the covariance of the temporal DFT coefficients from the sampled covariance sequence, then introduces the Whittle frequency-decoupling approximation, and finally relates the sampled spectrum to the continuous-time model through temporal aliasing.

\subsection{Temporal sampling, finite-record DFT covariance, and the Whittle model}
\label{sec:temporal-dft}

Let $\bm y_j=(Y(\bm s_1,t_j),\ldots,Y(\bm s_n,t_j))^\trans$ collect the observations at the $n$ spatial sites and at time $t_j=j\Delta t$, $j=0,\ldots,T-1$.  The unitary temporal discrete Fourier transform (DFT) of the observation sequence $\{\bm y_j\}_{j=0}^{T-1}$ is
\begin{equation}
 \bm d_k=\frac{1}{\sqrt T}\sum_{j=0}^{T-1}
 \bm y_j e^{-\ii\lambda_kj},
 \qquad
 \lambda_k=\frac{2\pi k}{T},
 \quad k=0,\ldots,T-1,
 \label{eq:dft-coefficient}
\end{equation}
where frequencies are interpreted modulo $2\pi$.  Write $\bm\Gamma_Y(r;\bm\theta)=\Cov(\bm y_{j+r},\bm y_j)$ for the sampled observation covariance at integer lag $r$.  Its entries are
\begin{equation}
 [\bm\Gamma_Y(r;\bm\theta)]_{ab}
 =C_Z(\bm s_a-\bm s_b,r\Delta t;\bm\theta)
 +v_{\mathrm n}(\bm\theta)\mathbb I(r=0)\mathbb I(a=b).
 \label{eq:sampled-observation-covariance}
\end{equation}
Directly applying \eqref{eq:dft-coefficient} to this covariance sequence gives the exact finite-record covariance blocks
\begin{equation}
 \bm F^{(T)}_{kk'}(\bm\theta)
 :={\mathbb E}\!\left[\bm d_k\bm d_{k'}^\herm\right]
 \label{eq:finite-record-dft-blocks}
 =\frac1T\sum_{j=0}^{T-1}\sum_{j'=0}^{T-1}
 \bm\Gamma_Y(j-j';\bm\theta)
 \,e^{-\ii\lambda_kj}\,e^{\ii\lambda_{k'}j'},
 \qquad k,k'=0,\ldots,T-1.
\end{equation}
Equation~\eqref{eq:finite-record-dft-blocks} serves as our exact finite-record starting point, which retains both the marginal covariance of each DFT vector and the cross-covariance between distinct frequencies.

To express these blocks in the frequency domain, assume that $\{C_Z(\bm h,r\Delta t;\bm\theta)\}_{r\in\mathbb Z}$ is absolutely summable for each spatial lag under consideration.  Its \emph{discrete-time temporal half-spectrum} is defined by the Fourier-series pair
\begin{align}
 H_{\mathrm d}(\bm h,\lambda;\bm\theta)
 &=\frac{1}{\sqrt{2\pi}}\sum_{r\in\mathbb Z}
 C_Z(\bm h,r\Delta t;\bm\theta)\,e^{-\ii\lambda r},
 \label{eq:sampled-half-spectrum-definition}\\
 C_Z(\bm h,r\Delta t;\bm\theta)
 &=\frac{1}{\sqrt{2\pi}}\int_{-\pi}^{\pi}
 H_{\mathrm d}(\bm h,\lambda;\bm\theta)\,e^{\ii\lambda r}
 \,\mathrm d\lambda.
 \label{eq:sampled-half-spectrum-inversion-main}
\end{align}
Here $\lambda\in[-\pi,\pi]$ is angular frequency in radians per sample, $H_{\mathrm d}$ is $2\pi$-periodic in $\lambda$, and the subscript $\mathrm d$ indicates discrete time.  On the covariance scale of the unitary data DFT, define the $n$ by $n$ matrix-valued sampled spectrum $\bm F(\lambda;\bm\theta)$ with elements
\begin{equation}
 [\bm F(\lambda;\bm\theta)]_{ab}
 :=\sqrt{2\pi}\,H_{\mathrm d}
 (\bm s_a-\bm s_b,\lambda;\bm\theta)
 +v_{\mathrm n}(\bm\theta)\mathbb I(a=b)\,,\qquad a, b=1,\dots,n\,.
 \label{eq:observation-spectrum-matrix}
\end{equation}
Equations~\eqref{eq:sampled-observation-covariance}, \eqref{eq:sampled-half-spectrum-inversion-main}, and \eqref{eq:observation-spectrum-matrix} then give
\begin{equation}
 \bm\Gamma_Y(r;\bm\theta)
 =\frac1{2\pi}\int_{-\pi}^{\pi}
 \bm F(\lambda;\bm\theta)\,e^{\ii\lambda r}\,\mathrm d\lambda.
 \label{eq:observation-covariance-inversion}
\end{equation}

The effect of observing only $T$ time points is described by the Dirichlet kernel $\mathcal D_T(x)=\sum_{j=0}^{T-1}e^{-\ii xj}$.  For concision in the next two displays, write $F_{ab}(\lambda)=[\bm F(\lambda;\bm\theta)]_{ab}$ and suppress the parameter argument of $\bm F^{(T)}_{kk'}(\bm\theta)$.  Substitution of \eqref{eq:observation-covariance-inversion} into \eqref{eq:finite-record-dft-blocks} then yields
\begin{equation}
 [\bm F^{(T)}_{kk'}]_{ab}
 =\frac1{2\pi T}\int_{-\pi}^{\pi}F_{ab}(\lambda)
 \,\mathcal D_T(\lambda_k-\lambda)
 \,\overline{\mathcal D_T(\lambda_{k'}-\lambda)}\,\mathrm d\lambda.
 \label{eq:exact-dft-window-covariance}
\end{equation}
For $k=k'$, let $\mathcal W_T(x)=|\mathcal D_T(x)|^2/(2\pi T)$ denote the normalized Fej\'er kernel, or finite-record spectral window \citep{brockwelldavis1991}, the exact marginal covariance can be further expressed as
\begin{align}
 [\bm F^{(T)}_{kk}]_{ab}
 &=\int_{-\pi}^{\pi}F_{ab}(\lambda)
 \,\mathcal W_T(\lambda_k-\lambda)\,\mathrm d\lambda
 \label{eq:fejer-smoothed-spectrum}\\
 &=\sum_{|r|<T}\Big(1-\frac{|r|}{T}\Big)\,
 [\bm\Gamma_Y(r;\bm\theta)]_{ab}\,e^{-\ii\lambda_kr}. \notag
\end{align}
Equation~\eqref{eq:fejer-smoothed-spectrum} shows that the marginal covariance of a finite-record DFT vector averages $\bm F(\lambda;\bm\theta)$ over a neighbourhood of $\lambda_k$, with weights given by the Fej\'er window.  The more general expression \eqref{eq:exact-dft-window-covariance} also retains cross-covariance between distinct DFT frequencies.  As $T$ increases, $\mathcal W_T$ remains nonnegative with unit integral and concentrates its mass increasingly near zero.  Therefore, at continuity points of $\bm F$, the weighted average in \eqref{eq:fejer-smoothed-spectrum} converges to $\bm F(\lambda_k;\bm\theta)$ of \eqref{eq:observation-spectrum-matrix}.  This concentration property motivates the frequency-wise covariance used by the Whittle approximation.

Using the short-hand notation
\begin{equation}
 \bm F_k(\bm\theta):=\bm F(\lambda_k;\bm\theta)\,,
 \label{eq:whittle-target-definition}
\end{equation}
the Whittle approximation \citep{whittle1953} neglects covariance between distinct Fourier frequencies and replaces the finite-record diagonal block $\bm F^{(T)}_{kk}(\bm\theta)$ by the spectral covariance $\bm F_k(\bm\theta)$.  For a positive interior frequency, the corresponding working model is
\begin{equation}
 \bm d_k\sim\CN\!\left(\bm0,\bm F_k(\bm\theta)\right),
 \qquad
 [\bm F_k(\bm\theta)]_{ab}
 =\sqrt{2\pi}\,H_{\mathrm d}
 (\bm s_a-\bm s_b,\lambda_k;\bm\theta)
 +v_{\mathrm n}(\bm\theta)\mathbb I(a=b)\,,
 \label{eq:Fk-final}
\end{equation}
where $\CN(\bm0,\bm F)$ denotes a proper complex Gaussian distribution with covariance $\mathbb E(\bm d\bm d^\herm)=\bm F$ and zero pseudo-covariance. The working model treats the positive-frequency vectors as independent.  The zero-frequency coefficient and, when $T$ is even, the Nyquist coefficient are real Gaussian.  Because the observed series is real-valued, the remaining coefficients of negative frequencies are fixed by $\bm d_{T-k}=\overline{\bm d_k}$.   \ref{app:temporal} states sufficient conditions for the Whittle limit and gives the circular model for which this block-diagonal frequency covariance is exact.

The Whittle approximation thus turns the temporally coupled covariance blocks $\bm F^{(T)}_{kk'}$ into separate spatial covariance problems indexed by $k$.  It remains to connect each sampled covariance $\bm F_k$ to the continuous-time spectral model.  Sampling at interval $\Delta t$ identifies continuous temporal frequencies that differ by integer multiples of $2\pi/\Delta t$, producing the usual aliasing sum.

To see this directly, insert the continuous inverse transform \eqref{eq:half-spectrum-inversion} into the sampled covariance at lag $r\Delta t$.  Partition the temporal-frequency axis into the intervals
\[
 I_m=\Big[
 \frac{-\pi+2\pi m}{\Delta t},
 \frac{\pi+2\pi m}{\Delta t}
 \Big),\qquad m\in\mathbb Z,
\]
and use $\tau=(\lambda+2\pi m)/\Delta t$ on $I_m$.  Since $e^{\ii2\pi mr}=1$ for integer $m$ and $r$, the result is
\begin{align}
 C_Z(\bm h,r\Delta t;\bm\theta)
 &=\frac{1}{\sqrt{2\pi}}\int_{-\pi}^{\pi}
 \Big(\frac{1}{\Delta t}\sum_{m\in\mathbb Z}
 H_{\mathrm c}\big(
 \bm h,\frac{\lambda+2\pi m}{\Delta t};\bm\theta
 \big)\Big)\,
 e^{\ii\lambda r}\,\mathrm d\lambda .
 \label{eq:main-folded-inverse}
\end{align}
Comparison with the discrete-time inversion formula \eqref{eq:sampled-half-spectrum-inversion-main} identifies the sampled half-spectrum as
\begin{equation}
 H_{\mathrm d}(\bm h,\lambda;\bm\theta)
 =\frac{1}{\Delta t}\sum_{m\in\mathbb Z}
 H_{\mathrm c}\big(
 \bm h,\frac{\lambda+2\pi m}{\Delta t};\bm\theta
 \big).
 \label{eq:main-half-spectrum-alias}
\end{equation}
Let $f_{\mathrm d}(\bm\omega,\lambda;\bm\theta)$ denote the unitary spatial Fourier transform of $H_{\mathrm d}(\bm h,\lambda;\bm\theta)$.  Taking the spatial Fourier transform of \eqref{eq:main-half-spectrum-alias} term by term gives
\begin{equation}
 f_{\mathrm d}(\bm\omega,\lambda;\bm\theta)
 :=\widehat H_{\mathrm d}(\bm\omega,\lambda;\bm\theta)
 =\frac{1}{\Delta t}\sum_{m\in\mathbb Z}
 f\big(\bm\omega,
 \frac{\lambda+2\pi m}{\Delta t};\bm\theta\big).
 \label{eq:sampled-joint-density}
\end{equation}
The interchange of the sum and integral and the term-wise spatial transform require the integrability conditions summarized in  \ref{app:temporal}.

\begin{proposition}[Sampled-frequency spatial symbol]
\label{prop:spatial-symbol}
Fix $\lambda\in[-\pi,\pi]$ and $\bm\theta\in\Theta$.  Suppose that $H_{\mathrm d}(\cdot,\lambda;\bm\theta)\in L^1(\R^d)$.  Define the convolution operator
\begin{equation}
 [\mathcal K^{\mathrm d}_{\lambda,\bm\theta}\,g](\bm s)
 :=\int_{\R^d}H_{\mathrm d}(\bm s-\bm s',\lambda;\bm\theta)
 \,g(\bm s')\,\mathrm d\bm s',
 \qquad g\in L^2(\R^d;\C).
 \label{eq:main-covariance-operator}
\end{equation}
Then $\mathcal K^{\mathrm d}_{\lambda,\bm\theta}$ is bounded on $L^2(\R^d;\C)$ and has sampled spectral multiplier
\begin{equation}
 \widehat{\mathcal K^{\mathrm d}_{\lambda,\bm\theta}g}(\bm\omega)
 =S_{\mathrm d}(\bm\omega,\lambda;\bm\theta)\,\widehat g(\bm\omega),
 \label{eq:main-multiplier-action}
\end{equation}
where
\begin{equation}
 \begin{split}
 S_{\mathrm d}(\bm\omega,\lambda;\bm\theta)
 &=(2\pi)^{d/2}\,f_{\mathrm d}(\bm\omega,\lambda;\bm\theta)\\
 &=\frac{1}{\Delta t}\sum_{m\in\mathbb Z}
 S_{\mathrm c}\big(
 \bm\omega,\frac{\lambda+2\pi m}{\Delta t};\bm\theta
 \big).
 \end{split}
 \label{eq:sampled-spectrum-definition}
\end{equation}
If $f_{\mathrm d}(\cdot,\lambda;\bm\theta)\in L^1(\R^d)$, then
\begin{equation}
 H_{\mathrm d}(\bm h,\lambda;\bm\theta)
 =\frac{1}{(2\pi)^d}\int_{\R^d}
 e^{\ii\bm\omega^\trans\bm h}
 \,S_{\mathrm d}(\bm\omega,\lambda;\bm\theta)\,\mathrm d\bm\omega
 \label{eq:hsgp-fourier-convention}
\end{equation}
at every continuity point of $H_{\mathrm d}$.
\end{proposition}

\begin{proof}
Young's inequality makes convolution by $H_{\mathrm d}$ a bounded map on $L^2$.  For $g\in L^1\cap L^2$, the convolution theorem gives $\widehat{\mathcal K^{\mathrm d}_{\lambda,\bm\theta}\,g} =(2\pi)^{d/2}f_{\mathrm d}\,\widehat g$.  Substitution of \eqref{eq:sampled-joint-density} and $S_{\mathrm c}=(2\pi)^{d/2}f$ yields equations \eqref{eq:main-multiplier-action} and \eqref{eq:sampled-spectrum-definition}. Density of the subspace $L^1\cap L^2$ then extends the multiplier identity to every $g\in L^2$.  Spatial Fourier inversion gives \eqref{eq:hsgp-fourier-convention}.
\end{proof}

\subsection{The HSGP covariance based on the sampled multiplier}
\label{sec:hsgp-covariance}

At frequency $\lambda_k$, the Whittle covariance in \eqref{eq:Fk-final} requires the spatial kernel $H_{\mathrm d}(\bm s_a-\bm s_b,\lambda_k;\bm\theta)$ for every pair of fitting sites.  When the model is specified through its joint spectrum or sampled multiplier, this spatial covariance can be constructed in two ways.
\begin{itemize}
  \item \textbf{\textit{Inverse-transform target calculation.}}
The sampled half-spectrum is recovered from $S_{\mathrm d}$ using \eqref{eq:hsgp-fourier-convention} and evaluated at the required spatial lags.  These values are then assembled into the dense target covariance $\bm F_k$.  The spatial inverse transform may be available analytically or evaluated numerically.

  \item \textbf{\textit{HSGP multiplier approximation.}}
The sampled multiplier $S_\mathrm{d}$ is evaluated at retained Laplacian eigenfrequencies.  Each evaluation becomes the coefficient of its corresponding Laplacian mode in a finite expansion on an extended bounded domain.  Evaluating that expansion at the fitting sites gives the low-rank covariance used in the Whittle likelihood and in spatial reconstruction.
\end{itemize}

The model of \citet[Eq.~(3)]{fuentes2008} provides an example with an analytic inverse transform, whereas the direct-spectral experiment in Section~\ref{sec:direct-spectral-experiment} uses numerical Hankel inversion to obtain its full-rank reference.  The remainder of this subsection derives the HSGP multiplier approximation.

The diagonal Laplacian construction used in this paper assumes that the sampled spatial multiplier is radial,
\begin{equation}
S_{\mathrm d}(\bm\omega,\lambda;\bm\theta)
=
\zeta_{\mathrm d}(\|\bm\omega\|,\lambda;\bm\theta),
\label{eq:radial-profile}
\end{equation}
where $\zeta_{\mathrm d}(r,\lambda;\bm\theta)$, $r\geq0$, is its radial profile.  This radial assumption defines the scope of the construction studied here.  For sufficiently regular $g$, the spatial Fourier transform satisfies
\[
\widehat{(-\Delta_{\bm s})g}(\bm\omega)
=
\|\bm\omega\|^2\,\widehat g(\bm\omega).
\]
Thus $\sqrt{-\Delta_{\bm s}}$ has Fourier multiplier $\|\bm\omega\|$, and
\[
\mathcal{F}_{\bm s}\{\zeta_{\mathrm d}
  (\sqrt{-\Delta_{\bm s}},\lambda;\bm\theta)g
\}(\bm\omega)
=
\zeta_{\mathrm d}
  (\|\bm\omega\|,\lambda;\bm\theta)\widehat g(\bm\omega)
=
S_{\mathrm d}(\bm\omega,\lambda;\bm\theta)\widehat g(\bm\omega)\,,
\]
where $\mathcal{F}_{\bm s}$ denotes the spatial Fourier transform.  Comparison with \eqref{eq:main-multiplier-action} gives the whole-space identity
\begin{equation}
\mathcal K^{\mathrm d}_{\lambda,\bm\theta}
=
\zeta_{\mathrm d}
  (\sqrt{-\Delta_{\bm s}},\lambda;\bm\theta).
\label{eq:main-functional-calculus}
\end{equation}
The HSGP construction of \citet{solinsarkka2020} transfers the scalar function $\zeta_{\mathrm d}$ from the whole-space Laplacian to a Laplacian on an extended bounded domain.  Let $c_r$ and $a_r>0$ be the centre and half-width of a spatial design box in coordinate $r=1,\ldots,d$, containing all observation and prediction sites.  Choose a boundary factor $c_B>1$ and define the extended domain
\begin{equation}
 \Omega=\prod_{r=1}^d[c_r-L_r,c_r+L_r],\qquad L_r=c_Ba_r.
 \label{eq:extended-domain}
\end{equation}
Let $\Delta_{\bm s}=\sum_{r=1}^d\frac{\partial^2}{\partial s_r^2}$ denote the spatial Laplacian.  Its Dirichlet eigenfunctions on $\Omega$ satisfy
\begin{equation}
 -\Delta_{\bm s}\phi_{\bm j}=\Lambda_{\bm j}\phi_{\bm j}
 \quad\text{in }\Omega,\qquad
 \phi_{\bm j}=0\quad\text{on }\partial\Omega,
 \label{eq:main-dirichlet-basis}
\end{equation}
where $\bm j\in\mathbb N^d$ is a multi-index with positive integer components and $\Lambda_{\bm j}>0$ is the associated eigenvalue.  The real functions $\phi_{\bm j}$ are normalized to form an orthonormal basis of $L^2(\Omega)$ with respect to spatial Lebesgue measure.  On the whole space, the negative Laplacian assigns eigenvalue $\|\bm\omega\|^2$ to spatial frequency $\bm\omega$.  This correspondence motivates calling $\rho_{\bm j}:=\sqrt{\Lambda_{\bm j}}$ the Laplacian eigenfrequency associated with $\phi_{\bm j}$. The explicit rectangular basis and its normalization are derived in  \ref{app:spatial}.  In the following, we write $-\Delta_\Omega$ for the nonnegative Dirichlet Laplacian defined by \eqref{eq:main-dirichlet-basis}.

Equation~\eqref{eq:main-functional-calculus} expresses the whole-space covariance operator as a scalar function of the negative Laplacian.  Following the HSGP construction of \citet{solinsarkka2020}, we replace the whole-space Laplacian by the Dirichlet Laplacian on $\Omega$ and define the resulting bounded-domain operator by
\[
 \mathcal A_{\lambda,\bm\theta,\Omega}
 :=\zeta_{\mathrm d}
 \bigl(\sqrt{-\Delta_\Omega},\lambda;\bm\theta\bigr).
\]
Since $-\Delta_\Omega$ is self-adjoint with eigenpairs $(\phi_{\bm j},\Lambda_{\bm j})$, spectral calculus gives
\[
 \mathcal A_{\lambda,\bm\theta,\Omega}\,\phi_{\bm j}
 =\zeta_{\mathrm d}(\rho_{\bm j},\lambda;\bm\theta)\,\phi_{\bm j}.
\]
The basis functions $\{\phi_{\bm j}\}$ therefore depend only on the extended domain, while the covariance parameters enter through the multiplier evaluations $\zeta_{\mathrm d}(\rho_{\bm j},\lambda;\bm\theta)$.

The bounded-domain operator remains infinite-dimensional.  HSGP obtains a finite representation by retaining a mode set $\mathcal J\subset\mathbb N^d$.  If $\Pi_{\phi}$ denotes the orthogonal projection onto $\operatorname{span}\,\{\phi_{\bm j}:\bm j\in\mathcal J\}$, the projected operator $\Pi_\phi\,\mathcal A_{\lambda,\bm\theta,\Omega}\,\Pi_{\phi}$ has kernel
\begin{equation}
 \widetilde H_{\mathrm d,\Omega,M}
 (\bm s,\bm s',\lambda;\bm\theta)
 :=\sum_{\bm j\in\mathcal J}
 \zeta_{\mathrm d}(\rho_{\bm j},\lambda;\bm\theta)\,
 \phi_{\bm j}(\bm s)\,\phi_{\bm j}(\bm s'),
 \qquad M=|\mathcal J|.
 \label{eq:hsgp-general-expansion}
\end{equation}
The kernel in \eqref{eq:hsgp-general-expansion} approximates $H_{\mathrm d}(\bm s-\bm s',\lambda;\bm\theta)$ through two steps: replacing the whole-space operator by its bounded-domain counterpart and projecting that operator onto the retained Laplacian modes. Since $S_{\mathrm d}$ is nonnegative, the weights in \eqref{eq:hsgp-general-expansion} are nonnegative and the truncated kernel is positive semidefinite.

To evaluate the approximation at the $n$ sites, define the $n\times M$ feature matrix $\bm\Phi$ and the $M\times M$ diagonal weight matrix $\bm W_k$ at each fixed temporal angular frequency $\lambda_k$ by
\begin{equation}
 \Phi_{i\bm j}=\phi_{\bm j}(\bm s_i),\qquad
 \bm W_k(\bm\theta)=\sqrt{2\pi}\,
 \diag\{\zeta_{\mathrm{d}}(\rho_{\bm j},\lambda_k;\bm\theta):\bm j\in\mathcal J\}.
 \label{eq:generic-hsgp-weights}
\end{equation}
For fitting sites $\bm s_a$ and $\bm s_b$, the $(a,b)$ entry of the low-rank matrix $\bm\Phi\bm W_k(\bm\theta)\bm\Phi^\trans$ is
\begin{align}
 [\bm\Phi\bm W_k(\bm\theta)\bm\Phi^\trans]_{ab}
 &=\sqrt{2\pi}\sum_{\bm j\in\mathcal J}
 \zeta_{\mathrm d}(\rho_{\bm j},\lambda_k;\bm\theta)
 \,\phi_{\bm j}(\bm s_a)\,\phi_{\bm j}(\bm s_b) \notag\\
 &=\sqrt{2\pi}\,
 \widetilde H_{\mathrm d,\Omega,M}
 (\bm s_a,\bm s_b,\lambda_k;\bm\theta).
 \label{eq:hsgp-matrix-entry}
\end{align}
Adding the measurement-noise variance to the diagonal gives the approximate observation covariance on the DFT covariance scale of \eqref{eq:Fk-final}:
\begin{equation}
 \widetilde{\bm F}_k(\bm\theta)
 =\bm\Phi\bm W_k(\bm\theta)\bm\Phi^\trans+v_{\mathrm{n}}(\bm\theta)\bm I_n.
 \label{eq:hsgp-frequency-covariance}
\end{equation}
Equations~\eqref{eq:generic-hsgp-weights}--\eqref{eq:hsgp-frequency-covariance} replace the spatial inverse transform in \eqref{eq:hsgp-fourier-convention} by $M$ multiplier evaluations and a low-rank matrix calculation.  The extension factor $c_B$ controls the distance to the artificial boundary, while $\mathcal J$ controls spectral truncation.   \ref{app:spatial} gives the basis formulas and approximation conditions.  Direct spectral weighting, a basis that is fixed as $\bm\theta$ and $\lambda_k$ change, and the cached projections that such a basis permits are established properties of HSGP \citep{solinsarkka2020}. They are inherited here at every temporal frequency.  The covariance in \eqref{eq:hsgp-frequency-covariance} uses one common rank $M=|\mathcal J|$ at every retained temporal frequency, which is standard HSGP applied frequency by frequency.  Section~ \ref{sec:adaptive-rank} replaces that common truncation by nested, frequency-specific ranks.

\subsection{Likelihood and spatial reconstruction}

The inverse-transform target calculation and the HSGP multiplier approximation use the same set of non-redundant temporal frequencies.  Defining the index sets $\mathcal K_+=\{1,\ldots,\lfloor(T-1)/2\rfloor\}$ and $\mathcal K_R=\{0\}$, and adding $T/2$ to $\mathcal K_R$ when $T$ is even, we write $\mathcal K=\mathcal K_+\cup\mathcal K_R$ and $K=|\mathcal K|$ for the retained non-redundant frequencies. Assume that the target covariances $\bm F_k(\bm\theta)$ are positive definite, and suppress the parameter arguments in the next two displays.  The one-sided negative log likelihood under the Whittle working model is
\begin{align}
 \mathcal L(\bm\theta)
 =\sum_{k\in\mathcal K_+}
 \big(\log|\bm F_k|+\bm d_k^\herm\bm F_k^{-1}\bm d_k+n\log\pi\big)
 +\frac12\sum_{k\in\mathcal K_R}
 \big(\log|\bm F_k|+\bm d_k^\trans\bm F_k^{-1}\bm d_k+n\log(2\pi)\big).
 \label{eq:whittle-likelihood}
\end{align}
The expression in \eqref{eq:whittle-likelihood} is the negative logarithm of the product density for the retained real and complex coefficients.  The corresponding HSGP Whittle objective is
\begin{equation}
 \widetilde{\mathcal L}(\bm\theta)
 =\sum_{k\in\mathcal K_+}
 \big(\log|\widetilde{\bm F}_k|
   +\bm d_k^\herm\widetilde{\bm F}_k^{-1}\bm d_k+n\log\pi\big)
 +\frac12\sum_{k\in\mathcal K_R}
 \big(\log|\widetilde{\bm F}_k|
   +\bm d_k^\trans\widetilde{\bm F}_k^{-1}\bm d_k+n\log(2\pi)\big).
 \label{eq:hsgp-whittle-likelihood}
\end{equation}

Fitting the target and HSGP covariance models gives
\begin{equation}
 \widehat{\bm\theta}=\argmin_{\bm\theta\in\Theta}\mathcal L(\bm\theta),
 \qquad
 \widetilde{\bm\theta}=\argmin_{\bm\theta\in\Theta}\widetilde{\mathcal L}(\bm\theta).
 \label{eq:exact-hsgp-estimators}
\end{equation}
The target estimate $\widehat{\bm\theta}$ uses spatial covariances obtained by an analytic full-rank or numerical-inversion calculation, whereas $\widetilde{\bm\theta}$ uses the HSGP covariances in \eqref{eq:hsgp-frequency-covariance}.  The displayed objectives assume a known zero mean.   \ref{app:inference} profiles transformed regression coefficients and gives the corresponding restricted Whittle objectives when a mean model is included.

Spatial reconstruction at unobserved locations follows the standard conditional Gaussian inferential framework for Gaussian process prediction \citep{rasmussenwilliams2006}, applied separately to the DFT coefficients at each retained temporal frequency. 
Let $\mathcal S=\{\bm s_1,\ldots,\bm s_n\}$ denote the fitting locations, and let  $\mathcal S_\ast=\{\bm s_1^\ast,\ldots,\bm s_{n_\ast}^\ast\}$ denote a collection of spatial reconstruction locations disjoint from $\mathcal S$. Also we use $\bm d_{\ast,k}$ to represent the corresponding DFT coefficient vector at the reconstruction locations $\mathcal S_\ast$.
For the target model, define the cross-covariance and reconstruction-site covariance matrices by
\begin{equation}
 \bm F_{\ast,k}=\Cov(\bm d_{\ast,k},\bm d_k),
 \qquad
 \bm F_{\ast\ast,k}=\operatorname{Var}(\bm d_{\ast,k}).
 \label{eq:reconstruction-covariance-blocks}
\end{equation}
Together with the fitting site covariance $\bm F_k$, these matrices specify the joint frequency-specific covariance structure of $\bm d_k$ and $\bm d_{\ast,k}$. All covariance blocks of the target-model are evaluated in the estimated parameter vector $\widehat{\bm\theta}$.

At each positive interior frequency, the Whittle working model assumes that $\bm d_k$ and $\bm d_{\ast,k}$ are jointly proper complex Gaussian. Conditioning on the observed coefficient vector $\bm d_k$ then yields
\begin{equation}
 \bm d_{\ast,k}\mid\bm d_k,\widehat{\bm\theta}
 \sim\CN\!\left(\widehat{\bm d}_{\ast,k},
 \bm V_{\ast,k}\right),
 \label{eq:frequency-prediction-law}
\end{equation}
with conditional mean and covariance
\begin{align}
 \widehat{\bm d}_{\ast,k}
 &=\bm F_{\ast,k}\bm F_k^{-1}\bm d_k,
 \label{eq:frequency-prediction}\\
 \bm V_{\ast,k}
 &=\bm F_{\ast\ast,k}
   -\bm F_{\ast,k}\bm F_k^{-1}
    \bm F_{\ast,k}^{\herm}.
 \label{eq:frequency-prediction-variance}
\end{align}
At zero frequency and, for even $T$, the Nyquist frequency, the corresponding conditional distributions are real Gaussian.

The HSGP reconstruction applies the same conditioning formulas, replacing the target covariance blocks with their low-rank approximations. Let $\bm\Phi$ and $\bm\Phi_\ast$ denote the matrices of retained Laplacian basis functions evaluated at the fitting and reconstruction locations, respectively. For the estimated HSGP parameter vector $\widetilde{\bm\theta}$, define $\widetilde v_{\mathrm n}=v_{\mathrm n}(\widetilde{\bm\theta})$. The resulting fitting-site covariance, cross-covariance,
and reconstruction-site covariance matrices are
\begin{align}
 \widetilde{\bm F}_k
 &=\bm\Phi\bm W_k(\widetilde{\bm\theta})\bm\Phi^{\trans}
   +\widetilde v_{\mathrm n}\bm I_n,\notag\\
 \widetilde{\bm F}_{\ast,k}
 &=\bm\Phi_\ast\bm W_k(\widetilde{\bm\theta})\bm\Phi^{\trans},
 \label{eq:hsgp-reconstruction-covariance-blocks}\\
 \widetilde{\bm F}_{\ast\ast,k}
 &=\bm\Phi_\ast\bm W_k(\widetilde{\bm\theta})
   \bm\Phi_\ast^{\trans}
   +\widetilde v_{\mathrm n}\bm I_{n_\ast}.\notag
\end{align}
Consequently, the HSGP conditional moments are 
\begin{align}
 \widetilde{\bm d}_{\ast,k}
 &=\widetilde{\bm F}_{\ast,k}\widetilde{\bm F}_k^{-1}\bm d_k,
 \label{eq:hsgp-frequency-prediction}\\
 \widetilde{\bm V}_{\ast,k}
 &=\widetilde{\bm F}_{\ast\ast,k}
   -\widetilde{\bm F}_{\ast,k}\widetilde{\bm F}_k^{-1}
    \widetilde{\bm F}_{\ast,k}^{\herm}.
 \label{eq:hsgp-frequency-prediction-variance}
\end{align}
For latent-field uncertainty, the term
$\widetilde v_{\mathrm n}\bm I_{n_{\ast}}$ is omitted from $\widetilde{\bm F}_{\ast\ast,k}$, whereas the noisy fitting site covariance $\widetilde{\bm F}_k$ is unchanged.  Because the fitting and reconstruction locations are disjoint, measurement noise does not contribute to the cross-covariance block.  The latent-field and noisy-observation conditional
means therefore coincide, although their conditional variances differ.

For a real-valued process, conjugate symmetry determines the conditional means at the redundant negative frequencies:
\begin{align}
 \widehat{\bm d}_{\ast,T-k}
 &=\overline{\widehat{\bm d}_{\ast,k}},\notag\\
 \widetilde{\bm d}_{\ast,T-k}
 &=\overline{\widetilde{\bm d}_{\ast,k}}.
 \label{eq:reconstruction-conjugate-symmetry}
\end{align}
After completion of the full frequency sequence, the target and HSGP reconstructions are respectively given by
\begin{align}
 \widehat{\bm y}_{\ast,j}
 &=\frac{1}{\sqrt T}\sum_{k=0}^{T-1}
   \widehat{\bm d}_{\ast,k}\,e^{\ii\lambda_kj},
 \qquad j=0,\ldots,T-1,
 \label{eq:time-domain-reconstruction}\\
 \widetilde{\bm y}_{\ast,j}
 &=\frac{1}{\sqrt T}\sum_{k=0}^{T-1}
   \widetilde{\bm d}_{\ast,k}\,e^{\ii\lambda_kj},
 \qquad j=0,\ldots,T-1.
 \label{eq:hsgp-time-domain-reconstruction}
\end{align}
When a mean model is fitted, conditioning is applied to the residual DFT coefficients and the fitted mean is subsequently restored.  The resulting uncertainty statements are conditional on the estimated covariance and mean parameters.   \ref{app:inference} gives the corresponding time-domain variance reconstruction.

\subsection{Frequency-adaptive selection of the spatial rank}
\label{sec:adaptive-rank}

The HSGP construction in Section~\ref{sec:hsgp-covariance} uses a common spatial rank at every temporal frequency.  Such a rank must be large enough for the frequency whose sampled multiplier contains the finest spatial structure.  When spatial complexity varies with $\lambda_k$, applying that same rank throughout can retain more Laplacian modes than other frequencies need.  This motivates a frequency-adaptive construction in which the domain and ordered basis remain fixed, but the retained rank $M_k$ may vary with $k$.

To compare possible ranks, first order the master Laplacian basis on $\Omega$ by increasing eigenvalue.  Let $\mathcal C=\{m_\ell\}_{\ell=1}^L$ be the candidate ranks with $m_1<\cdots<m_L$.  For each $m\in\mathcal C$, the matrix $\bm\Phi_m$ contains the first $m$ columns of the master basis, and $\bm W_{k,m}(\bm\theta)$ is the corresponding leading block of the weight matrix in \eqref{eq:generic-hsgp-weights}.  Thus, if $m_\ell<m_r$, every mode retained at rank $m_\ell$ is also retained at rank $m_r$; the larger approximation differs only by its additional modes.  At fitting sites, rank $m$ gives the covariance
\begin{equation}
 \bm Q_{k,m}(\bm\theta)
 =\bm\Phi_m\bm W_{k,m}(\bm\theta)\bm\Phi_m^\trans
  +v_{\mathrm n}(\bm\theta)\bm I_n.
 \label{eq:candidate-hsgp-covariance}
\end{equation}
This common ordering makes the candidate approximations nested.

\paragraph{Internal reference and covariance discrepancy.}
The approximation at a candidate rank requires a reference covariance against which its loss of spatial detail can be measured.  Denote this reference at frequency $k$ by $\bm P_k(\bm\theta)$.  For the operational procedure, we construct $\bm P_k(\bm\theta)$ internally from the same nested HSGP basis.  Let $M_{\mathrm{ref}}\geq m_L$ be a rank at which the fitting-site HSGP covariance has stabilized as the basis is enlarged, and set
\begin{equation}
 \bm P_k(\bm\theta)
 =\bm Q_{k,M_{\mathrm{ref}}}(\bm\theta).
 \label{eq:adaptive-internal-reference}
\end{equation}
The stabilization calculation is described in  \ref{app:adaptive-kl-details}.  An analytic or high-accuracy numerical covariance may instead be used as an external reference in benchmark studies, as discussed at the end of this subsection.

The Whittle model treats $\bm d_k$ as real Gaussian at the endpoint frequencies and as proper complex Gaussian at the interior positive frequencies.  For these two cases, the Kullback-Leibler (KL) divergence from a reference covariance $\bm P$ to a candidate covariance $\bm Q$ is
\begin{equation}
 \mathcal D_k(\bm P\Vert\bm Q)
 =c_k\big(
  \log|\bm Q|-\log|\bm P|
  +\operatorname{tr}(\bm Q^{-1}\bm P)-n
 \big),
 \qquad
 c_k=\begin{cases}
 1/2,&k\in\mathcal K_R,\\
 1,&k\in\mathcal K_+.
 \end{cases}
 \label{eq:adaptive-gaussian-kl}
\end{equation}
We measure the discrepancy per fitting site,
\begin{equation}
 e_{k,m}(\bm\theta)
 =\frac{1}{n}\mathcal D_k\!\big(
   \bm P_k(\bm\theta)\Vert\bm Q_{k,m}(\bm\theta)\big).
 \label{eq:adaptive-normalized-kl}
\end{equation}
This quantity is the expected increase in the frequency-wise negative log density per fitting site when the rank-$m$ covariance replaces the internal reference.  It therefore provides a common scale on which ranks can be compared across temporal frequencies.

\paragraph{Frequency-specific rank choice.}
Given a tolerance $\epsilon_{\mathrm{KL}}$, the selected rank is the smallest candidate whose discrepancy, together with those of all larger candidates, remains below the tolerance:
\begin{equation}
 \widehat M_k(\bm\theta)
 =\min\left\{m_\ell\in\mathcal C:
   \max_{r=\ell,\ldots,L} e_{k,m_r}(\bm\theta)
   \leq\epsilon_{\mathrm{KL}}\right\}.
 \label{eq:adaptive-rank-rule}
\end{equation}
The maximum over larger candidates protects the selection against small numerical departures from monotone error decay.  The tolerance is specified before held-site responses are examined and is applied separately at every temporal frequency.

Both the multiplier weights and the internal reference depend on $\bm\theta$, so the ranks selected before parameter estimation need not remain adequate after fitting.  Algorithm~\ref{alg:adaptive-hsgp} addresses this dependence in three stages.  A pilot fit with a common rank $M_{\mathrm{init}}$ first supplies an initial parameter estimate.  The rule in \eqref{eq:adaptive-rank-rule} then selects a rank at each frequency, and the Whittle likelihood is fitted using those frequency-specific ranks.  After fitting, the same rule is evaluated again at the updated parameter estimate. If it selects a rank larger than the one used at any frequency, that rank is increased before the likelihood is fitted again.

\par\medskip \refstepcounter{algorithm}\label{alg:adaptive-hsgp} \noindent\rule{\textwidth}{0.8pt} \noindent\textbf{Algorithm~\thealgorithm}\quad Frequency-adaptive HSGP using the KL criterion \par\noindent\rule{\textwidth}{0.4pt} {\small
\begin{algorithmic}[1]
\Statex \textbf{Data:} Observed vectors
$\{\bm y_j\in\mathbb R^n:j=0,\ldots,T-1\}$ at fitting sites $\mathcal S$;
reconstruction sites $\mathcal S_\ast$.
%; mean specification, when included.
\Statex \textbf{Model and basis:} Sampled multiplier
$S_{\mathrm d}(\bm\omega,\lambda;\bm\theta)$, parameter space $\Theta$,
extended domain $\Omega$, and an ordered master Laplacian basis through rank
$M_{\mathrm{ref}}$.
\Statex \textbf{Rank controls:} Pilot rank $M_{\mathrm{init}}$; candidates
$\mathcal C=\{m_1<\cdots<m_L\}$ with $m_L\leq M_{\mathrm{ref}}$; tolerance
$\epsilon_{\mathrm{KL}}$; maximum number of correction fits $I_{\max}$.
\State Apply the temporal DFT to obtain
$\{\bm d_k:k\in\mathcal K\}$, evaluate the basis at $\mathcal S$ and
$\mathcal S_\ast$, and cache the fitting-site feature projections.
\State Fit the common-rank pilot model to obtain
$\widehat{\bm\theta}^{(0)}$.
\State At $\widehat{\bm\theta}^{(0)}$, form the internal references in
\eqref{eq:adaptive-internal-reference} and set
$M_k^{(0)}=\widehat M_k(\widehat{\bm\theta}^{(0)})$ for every
$k\in\mathcal K$ using \eqref{eq:adaptive-rank-rule}.
\For{$i=1,\ldots,I_{\max}$}
 \State Fit the Whittle likelihood with ranks $\{M_k^{(i-1)}\}$ to obtain
 $\widehat{\bm\theta}^{(i)}$.
 \State Recompute $\widehat M_k(\widehat{\bm\theta}^{(i)})$ for every
 $k\in\mathcal K$.
 \If{$\widehat M_k(\widehat{\bm\theta}^{(i)})\leq M_k^{(i-1)}$ for every $k$}
  \State Set $\widehat M_k=M_k^{(i-1)}$ and
  $\widehat{\bm\theta}_{\mathrm{ad}}=\widehat{\bm\theta}^{(i)}$, then
  terminate the loop.
 \Else
  \State Set
  $M_k^{(i)}=\max\{M_k^{(i-1)},
  \widehat M_k(\widehat{\bm\theta}^{(i)})\}$ for every $k$.
 \EndIf
\EndFor
\If{no verified profile has been obtained after $I_{\max}$ correction fits}
 \State Report the rank profile as unresolved and stop.
\EndIf
\State Use $\widehat{\bm\theta}_{\mathrm{ad}}$ and
$\{\widehat M_k:k\in\mathcal K\}$ in
\eqref{eq:hsgp-frequency-prediction} and
\eqref{eq:hsgp-time-domain-reconstruction} to reconstruct at
$\mathcal S_\ast$.
\Statex \textbf{Output:} Fitted parameters
$\widehat{\bm\theta}_{\mathrm{ad}}$; verified ranks
$\{\widehat M_k:k\in\mathcal K\}$; reconstructed time series and requested
conditional covariances at $\mathcal S_\ast$.
\end{algorithmic}
} \noindent\rule{\textwidth}{0.8pt} \par\medskip

The internal reference is evaluated only when the ranks are selected or reassessed.  All objective evaluations within the variable-rank fits, together with the final reconstruction, use the selected ranks $\widehat M_k$.  Thus, compared with retaining $M_{\mathrm{ref}}$ at every frequency throughout optimization and reconstruction, the adaptive procedure can reduce the repeated calculations when many selected ranks are appreciably smaller than $M_{\mathrm{ref}}$.  The cost of constructing the internal references is incurred only at the comparatively few rank-assessment steps. 

Although the rank criteria do not use the observed DFT coefficients directly, the candidate and reference covariances are evaluated at parameters $\widehat{\bm{\theta}}$ estimated from those coefficients. The resulting rank profile $\{\widehat{M}_k\}_{k\in\mathcal{K}}$ therefore adapts to the covariance structure estimated from the observed record and is reassessed whenever the parameter estimate changes. The procedure is thus model-based and adaptive to the fitting, with held-site responses reserved for out-of-sample evaluation.

\paragraph{Trace-based sufficient criterion.}
Evaluating the KL divergence in \eqref{eq:adaptive-gaussian-kl} requires factorizations of both the candidate covariance and the internal reference. We also consider an alternative criterion that uses the nested basis to avoid factorizing the internal reference for every candidate comparison.  For a candidate rank $m<M_{\mathrm{ref}}$, the modes between $m$ and $M_{\mathrm{ref}}$ contribute the omitted covariance
\begin{align}
 \bm R_{k,m}(\bm\theta)
 &:=\bm P_k(\bm\theta)-\bm Q_{k,m}(\bm\theta) \notag\\
 &=\sum_{j=m+1}^{M_{\mathrm{ref}}}
   w_{j,k}(\bm\theta)\,
   \bm\phi_j\,\bm\phi_j^\trans \succeq \bm 0,
 \label{eq:trace-tail-covariance}
\end{align}
where $\bm\phi_j$ is the $j$th master-basis column at the fitting sites and $w_{j,k}$ is the corresponding diagonal entry of $\bm W_k$.  The scalar
\begin{equation}
 t_{k,m}(\bm\theta)
 =\operatorname{tr}\!\big(
   \bm Q_{k,m}^{-1}\bm R_{k,m}\big)
 \label{eq:trace-tail-action}
\end{equation}
measures the omitted covariance relative to the covariance already represented at rank $m$.   \ref{app:adaptive-kl-details} shows that
\begin{equation}
 e_{k,m}(\bm\theta)
 \leq \frac{c_k}{n}
 \Big(t_{k,m}(\bm\theta)
 -\log\big(1+t_{k,m}(\bm\theta)\big)\Big)
 =:b_{k,m}^{\mathrm{tr}}(\bm\theta).
 \label{eq:trace-kl-bound}
\end{equation}
The trace-certified rank is consequently
\begin{equation}
 \widehat M_k^{\mathrm{tr}}(\bm\theta)
 =\min\left\{m_\ell\in\mathcal{C}:
 b_{k,m_\ell}^{\mathrm{tr}}(\bm\theta)
 \leq\epsilon_{\mathrm{KL}},\quad \ell=1,\ldots,L\right\}.
 \label{eq:trace-rank-rule}
\end{equation}
Because \eqref{eq:trace-kl-bound} is an upper bound, every accepted rank meets the KL tolerance, although the rule can retain more modes than the direct KL comparison.

\par\medskip \refstepcounter{algorithm}\label{alg:trace-adaptive-hsgp} \noindent\rule{\textwidth}{0.8pt} \noindent\textbf{Algorithm~\thealgorithm}\quad Frequency-adaptive HSGP using the trace criterion \par\noindent\rule{\textwidth}{0.4pt} {\small
\begin{algorithmic}[1]
\Statex \textbf{Data:} Observed vectors
$\{\bm y_j\in\mathbb R^n:j=0,\ldots,T-1\}$ at fitting sites $\mathcal S$;
reconstruction sites $\mathcal S_\ast$.
%; mean specification, when included.
\Statex \textbf{Model and basis:} Sampled multiplier
$S_{\mathrm d}(\bm\omega,\lambda;\bm\theta)$, parameter space $\Theta$,
extended domain $\Omega$, and an ordered master Laplacian basis through rank
$M_{\mathrm{ref}}$.
\Statex \textbf{Rank controls:} Pilot rank $M_{\mathrm{init}}$; candidates
$\mathcal C=\{m_1<\cdots<m_L\}$ with $m_L\leq M_{\mathrm{ref}}$; tolerance
$\epsilon_{\mathrm{KL}}$; maximum number of correction fits $I_{\max}$.
\State Compute $\{\bm d_k:k\in\mathcal K\}$, the basis evaluations, and the
cached projections as in Algorithm~\ref{alg:adaptive-hsgp}. Then fit the pilot
model to obtain $\widehat{\bm\theta}^{(0)}$.
\State At $\widehat{\bm\theta}^{(0)}$, set
$M_k^{(0)}=\widehat M_k^{\mathrm{tr}}
(\widehat{\bm\theta}^{(0)})$ for every $k\in\mathcal K$ using
\eqref{eq:trace-rank-rule}.
\For{$i=1,\ldots,I_{\max}$}
 \State Fit the Whittle likelihood with ranks $\{M_k^{(i-1)}\}$ to obtain
 $\widehat{\bm\theta}^{(i)}$.
 \State Recompute
 $\widehat M_k^{\mathrm{tr}}(\widehat{\bm\theta}^{(i)})$ for every
 $k\in\mathcal K$.
 \If{$\widehat M_k^{\mathrm{tr}}(\widehat{\bm\theta}^{(i)})
 \leq M_k^{(i-1)}$ for every $k$}
  \State Set $\widehat M_k=M_k^{(i-1)}$ and
  $\widehat{\bm\theta}_{\mathrm{tr}}=\widehat{\bm\theta}^{(i)}$, then
  terminate the loop.
 \Else
  \State Set $M_k^{(i)}=\max\{M_k^{(i-1)},
  \widehat M_k^{\mathrm{tr}}(\widehat{\bm\theta}^{(i)})\}$ for every $k$.
 \EndIf
\EndFor
\If{no verified profile has been obtained after $I_{\max}$ correction fits}
 \State Report the rank profile as unresolved and stop.
\EndIf
\State Reconstruct at $\mathcal S_\ast$ using
$\widehat{\bm\theta}_{\mathrm{tr}}$ and
$\{\widehat M_k:k\in\mathcal K\}$ as in Algorithm~\ref{alg:adaptive-hsgp}.
\Statex \textbf{Output:} Fitted parameters
$\widehat{\bm\theta}_{\mathrm{tr}}$; verified ranks
$\{\widehat M_k:k\in\mathcal K\}$; reconstructed time series and requested
conditional covariances at $\mathcal S_\ast$.
\end{algorithmic}
} \noindent\rule{\textwidth}{0.8pt} \par\medskip

The trace criterion changes only the rank-assessment step.  For each candidate, it factorizes $\bm Q_{k,m}$ and applies triangular solves to the omitted basis columns, without factorizing $\bm P_k$.  It can therefore shorten the rank audit relative to Algorithm~\ref{alg:adaptive-hsgp}, particularly when the number of fitting sites or candidate ranks is large.  The subsequent variable-rank optimization and reconstruction are unchanged.  Because \eqref{eq:trace-kl-bound} is a sufficient bound, Algorithm~ \ref{alg:trace-adaptive-hsgp} can retain more modes than the direct KL rule. Section~\ref{sec:ghcnh-wind} compares the two criteria on a European station network.

When an analytic or high-accuracy numerical covariance $\bm F_k$ is available in a benchmark, the selected ranks can also be evaluated against the external target by setting
\begin{equation}
 \bm P_k(\bm\theta)=\bm F_k(\bm\theta).
 \label{eq:adaptive-target-reference}
\end{equation}
This external comparison assesses finite-domain and truncation error together. It is used for validation in the simulation studies and is not required by either operational algorithm.  \ref{app:adaptive-kl-details} gives the reference-stability rule, Cholesky evaluation of the KL discrepancy, trace-bound derivation, and procedures for unresolved candidate or reference grids.

\subsection{Cached computation and leading cost}
\label{sec:cached-calculation}

The fixed master dictionary allows geometry-dependent and data-dependent quantities to be reused throughout parameter optimization.  Let $M_{\max}=\max\{m_L,M_{\mathrm{ref}}\}=M_{\mathrm{ref}}$ denote the largest basis rank required by either candidate fitting or internal-reference construction.  Let $\bm\Phi$ contain the first $M_{\max}$ columns of the master feature matrix. The following quantities can be cached:
\begin{equation}
\bm G=\bm\Phi^\trans\bm\Phi,\qquad
\bm b_k=\bm\Phi^\trans\bm d_k,\qquad
u_k=\bm d_k^\herm\bm d_k.
\label{eq:cached-gram}
\end{equation}
% The relevant leading blocks of $\bm G$ and $\bm b_k$ are shared by every candidate rank and by the internal reference, while $u_k$ is computed once for each DFT coefficient vector. We can then apply the determinant lemma and Woodbury identity to reduce the likelihood calculation to feature space. Matrix inverses are implemented through linear solves, normally using Cholesky factors.
The relevant leading blocks of $\bm G$ and $\bm b_k$ are shared by every
candidate rank and by the internal reference, while $u_k$ is computed once for
each DFT coefficient vector.  The matrix determinant lemma and the Woodbury
identity then reduce the log-determinant and quadratic-form calculations in the
Whittle objective to feature space, as described by
\citet[ A.3]{rasmussenwilliams2006}.  In computation, the required linear systems are
solved using Cholesky factors without explicitly forming matrix inverses.

The temporal DFT and the quantities in \eqref{eq:cached-gram} are computed once.  For a real-valued record of length $T$, the number of nonredundant DFT frequencies is proportional to $T$.  Table~\ref{tab:computational-cost} summarizes the leading factorization cost per objective evaluation and the additional model-specific storage when all frequency projections $\bm b_k$ are cached.

\begin{table}[H]
\centering
\small
\caption{Leading factorization cost per objective evaluation and additional model-specific storage with cached frequency projections.}
\label{tab:computational-cost}
\begin{tabular}{@{}lcc@{}}
\toprule
Spatial covariance calculation
& \shortstack{Cost per objective evaluation}
& Additional storage \\
\midrule
Analytic full-rank
& $O(Tn^3)$
& $O(n^2)$ \\
Fixed-rank HSGP in feature space
& $O(TM^3)$
& $O(nM+M^2+TM)$ \\
Frequency-adaptive HSGP in feature space
& $O\!\left(\sum_{k\in\mathcal K}M_k^3\right)$
& $O(nM_{\max}+M_{\max}^2+TM_{\max})$ \\
\bottomrule
\end{tabular}
\end{table}

The frequency-adaptive row gives the recurring cost after the rank profile has been selected.  A fixed-rank HSGP that retains $M_{\mathrm{ref}}$ modes at every frequency incurs $O(TM_{\mathrm{ref}}^3)$ per objective evaluation.  By contrast, the adaptive cost is $O\big(\sum_{k\in\mathcal K}M_k^3\big)$.  It is therefore smaller when the selected ranks satisfy $M_k<M_{\mathrm{ref}}$ at most frequencies, with the largest savings arising when those reductions are substantial.  The pilot fit, internal-reference construction, rank audits, and correction fits add occasional costs.  The adaptive procedure reduces total computation when the savings accumulated over likelihood evaluations and reconstruction exceed this additional rank-selection work.

\section{Numerical experiments}
\label{sec:numerics}

\subsection{Benchmark spectral models and evaluation}
\label{sec:numerical-setup}
\label{sec:fuentes}

The numerical experiments use two complementary spectral constructions.  The first is the stationary nonseparable model of \citet[Eq.~(3)]{fuentes2008}, with continuous space-time spectral density
\begin{equation}
 f(\bm\omega,\tau;\bm\theta)=\varsigma
 \left[
 \alpha^2\beta^2+\beta^2\|\bm\omega\|^2
 +\alpha^2\tau^2+\epsilon\|\bm\omega\|^2\tau^2
 \right]^{-\nu}.
 \label{eq:fuentes-spectrum}
\end{equation}
Here $\alpha>0$ and $\beta>0$ determine spatial and temporal scales, respectively, while $\epsilon>0$ controls the space-time interaction.  The parameter $\nu>\max(\frac{d}{2},\frac{1}{2})$ ensures integrability.  The positive scale $\varsigma$ is uniquely determined by requiring the marginal latent variance to equal $v_{\mathrm f}$.  Thus $f$ is normalized as a spectral density with prescribed total spectral mass.  It is not normalized as a probability density.

For a fixed temporal frequency $\tau$, the half-spectrum is obtained by spatially inverting the joint spectral density in \eqref{eq:fuentes-spectrum}.  Let
\[
A(\tau)=\alpha^2(\beta^2+\tau^2),\qquad
B(\tau)=\beta^2+\epsilon\tau^2,\qquad
\eta=\nu-\frac{d}{2}.
\]
The spatial inversion is available analytically and gives
\begin{equation}
H_{\mathrm c}(\bm h,\tau;\bm\theta)
\propto
p(\tau)\,
\Matern_{\eta}\!\left(
\kappa(\tau)\|\bm h\|
\right),
\qquad
p(\tau)
=
A(\tau)^{\frac d2-\nu}B(\tau)^{-\frac d2},
\label{eq:fuentes-half-spectrum}
\end{equation}
where $\Matern_{\eta}$ denotes the Mat\'ern correlation function of order $\eta$ and
\begin{equation}
\kappa(\tau)
=
\Big(\frac{A(\tau)}{B(\tau)}\Big)^{1/2}
=
\alpha
\Big(
\frac{\beta^2+\tau^2}
     {\beta^2+\epsilon\tau^2}
\Big)^{1/2}
\label{eq:fuentes-range}
\end{equation}
is the frequency-dependent inverse spatial range.
% The proportionality sign in \eqref{eq:fuentes-half-spectrum} suppresses only this uniquely determined overall scale and the Fourier-transform constants.

Both the amplitude $p(\tau)$ and the spatial range $\kappa(\tau)$ may therefore vary with temporal frequency.  When $\epsilon=1$, $\kappa(\tau)=\alpha$ and the space-time model becomes separable.  The analytic expression \eqref{eq:fuentes-half-spectrum} supplies the full-rank spatial covariance used as the reference calculation, whereas the HSGP construction evaluates the radial profile of the corresponding spectral multiplier at the retained Laplacian eigenfrequencies. The two calculations consequently provide distinct computational representations of the same frequency-specific spatial covariance.
% Temporal sampling introduces an alias sum, for which both representations use the same normalization and truncation. 
We provide further details in  \ref{app:fuentes}.

The second construction, introduced in Section~\ref{sec:direct-spectral-experiment}, specifies the sampled spatial multiplier directly.  Its half-spectrum does not have a convenient spatial inverse transform in closed form, so a high-accuracy numerical inversion provides the reference covariance.  This benchmark tests the main computational setting of interest: HSGP can use the multiplier directly, whereas the full-rank calculation requires repeated spatial integration.

The interpretation of the benchmarks separates the temporal Whittle approximation from the spatial HSGP approximation.  For ordinary finite records, the Whittle likelihood neglects cross-frequency dependence.  In the synthetic studies described as circular, independent Fourier coefficients are generated from the frequency-wise Gaussian model with imposed conjugate symmetry, and the inverse DFT produces the time-domain record.  Frequency decoupling is therefore exact for these simulations, so differences among spatial methods arise from covariance representation, parameter estimation, and numerical optimization.

The spatial studies compare held-site reconstruction error and, where relevant, observed runtime.  The Fuentes experiment can evaluate the KL rank criterion against the analytic full-rank covariance.  In models without that analytic calculation, Algorithm~\ref{alg:adaptive-hsgp} instead uses a stabilized high-rank HSGP covariance as its internal reference.  Rank selection and model fitting use only the fitting locations, and held-site responses are reserved for evaluation.

\subsection{Synthetic experiments}

\subsubsection{Analytic-reference study: the Fuentes model}
\label{sec:fuentes-synthetic-experiment}

This experiment evaluates the spatial HSGP approximation for the Fuentes spectral model \eqref{eq:fuentes-spectrum}.  Because the spatial inverse transform is available analytically, the model supplies an analytic full-rank covariance reference at every temporal frequency.

The simulation used a $12\times12$ spatial grid, $T=192$ regularly spaced time points, $n=120$ fitting sites, and $n_\ast=24$ whole-site holdouts.  We generated 20 independent realizations of each configuration.  Within a realization, every method used the same field and the same fitting-testing split, which makes the method comparisons paired.  The signal and noise variances were fixed at $v_{\mathrm f}=1$ and $v_{\mathrm n}=0.05$, respectively. The smoothness was $\nu=2.5$, the sampling interval was one model-time unit, and the sampled spectrum retained temporal aliases $q=-6,\ldots,6$.  The domain-extension factor $c_B$ defined in \eqref{eq:extended-domain} was 6 for Configuration A and 2.5 for Configuration B.  These values and the fixed ranks reported below were selected jointly by a preliminary covariance audit against the analytic full-rank covariance over a prespecified parameter and location grid.  The audit used neither simulated responses nor held-site errors. Configuration A used covariance model parameters $(\alpha,\beta,\epsilon)=(1.5,0.35,0.35)$, whereas Configuration B used $(\alpha,\beta,\epsilon)=(5,0.35,0.35)$ and therefore represented the shorter-range, more spatially demanding setting.

The analytic-reference and HSGP likelihoods were optimized independently for each realization by the same multistart Nelder-Mead procedure, with matched initial values and stopping rules.  Reconstruction was then evaluated under the five combinations in Block~I of Table~\ref{tab:fuentes-prediction-rmse}.  The first two rows hold the parameters at their generating value and compare the two spatial covariance calculations.  The remaining rows compare the fitted analytic-reference and HSGP parameter estimates, both before and after the HSGP covariance approximation is introduced.  This design separates errors from parameter estimation from errors due to the finite spatial basis.  The fixed HSGP bases contain $M=400$ and $M=576$ functions for Configurations A and B, respectively.

Table~\ref{tab:fuentes-prediction-rmse} addresses two questions.  Block~I separates parameter-estimation and covariance-approximation effects.  Block~II then studies whether frequency-specific ranks can recover more of the analytic full-rank accuracy than a single fixed rank.  All intervals use a paired bootstrap over the same 20 realizations.

\begin{table}[H]
\centering
\small
\caption{Mean held-site latent-field RMSE over 20 paired reconstruction
replicates.  Entries are the mean $\pm$ a 95\% paired-bootstrap half-width.
Block~I separates parameter-estimation and covariance-representation effects.
Block~II compares fixed-rank and frequency-adaptive HSGP at the
replicate-specific analytic-reference estimate.}
\label{tab:fuentes-prediction-rmse}
\begin{tabular}{@{}llcc@{}}
\toprule
Parameter source & Reconstruction covariance & Configuration A & Configuration B \\
\midrule
\multicolumn{4}{@{}l}{\textit{I. Parameter and covariance representation}}\\
Oracle $\bm\theta_0$ & Analytic full-rank
  & $0.0757\pm 0.0007$ & $0.1917\pm 0.0017$ \\
Oracle $\bm\theta_0$ & HSGP
  & $0.0783\pm0.0007$ & $0.1962\pm 0.0017$ \\
Analytic-reference fit & Analytic full-rank
  & $0.0758\pm 0.0006$ & $0.1917\pm 0.0016$ \\
HSGP fit & Analytic full-rank
  & $0.0760\pm0.0006$ & $0.1918\pm 0.0016$ \\
HSGP fit & HSGP
  & $0.0781\pm 0.0007$ & $0.1975\pm 0.0017$ \\
\addlinespace
\multicolumn{4}{@{}l}{\textit{II. Frequency-specific rank validation}}\\
Analytic-reference fit & Analytic full-rank
  & $0.0758\pm0.0006$ & $0.1917\pm0.0016$ \\
Analytic-reference fit & Fixed-rank HSGP
  & $0.0784\pm0.0007$ & $0.1962\pm0.0016$ \\
Analytic-reference fit & Frequency-adaptive HSGP
  & $0.0773\pm0.0007$ & $0.1926\pm0.0016$ \\
\bottomrule
\end{tabular}
\end{table}

Block~I shows that the five schemes occupy a narrow error range within each configuration.  When both fitted parameter vectors were evaluated with the analytic covariance, their mean RMSEs differed by less than $0.3\%$.  By comparison, replacing the analytic reconstruction covariance by its finite HSGP representation increased mean RMSE by about $2.3\%$--$3.4\%$ across the matched comparisons in Block~I.  The finite spatial representation therefore accounts for most of the difference between the two complete fitting-and-reconstruction procedures.

The close reconstruction errors generally do not imply identical fitted models.  The HSGP fits showed an upward shift in $\alpha$, $\beta$, and $\epsilon$ relative to the analytic-reference fits. Appendix Table~\ref{tab:parameter-recovery} reports the corresponding parameter summaries.

\paragraph{Validation of frequency-specific ranks.}
The analytic covariance also permits direct validation of the rank criterion, with $\bm P_k=\bm F_k$ in \eqref{eq:adaptive-target-reference}.  This is the external validation of Section~\ref{sec:adaptive-rank} and does not use the internal reference.  We applied this comparison to the same 20 estimation and prediction pairs used above.  For each replicate, the ranks were selected at its analytic-reference parameter estimate using fitting-site covariance blocks.  Candidate bases used $M=J^2$, where $J$ is the number of retained modes in each spatial coordinate, for $J$ ranging from $6$ to $40$, with $\epsilon_{\mathrm{KL}}=0.01$ per fitting site.  Analytic, fixed-rank HSGP, and frequency-adaptive HSGP reconstruction were then evaluated at the same replicate-specific parameter estimate, so their differences isolate the spatial covariance representation.

\begin{figure}[H]
\centering
\includegraphics[width=0.98\textwidth]{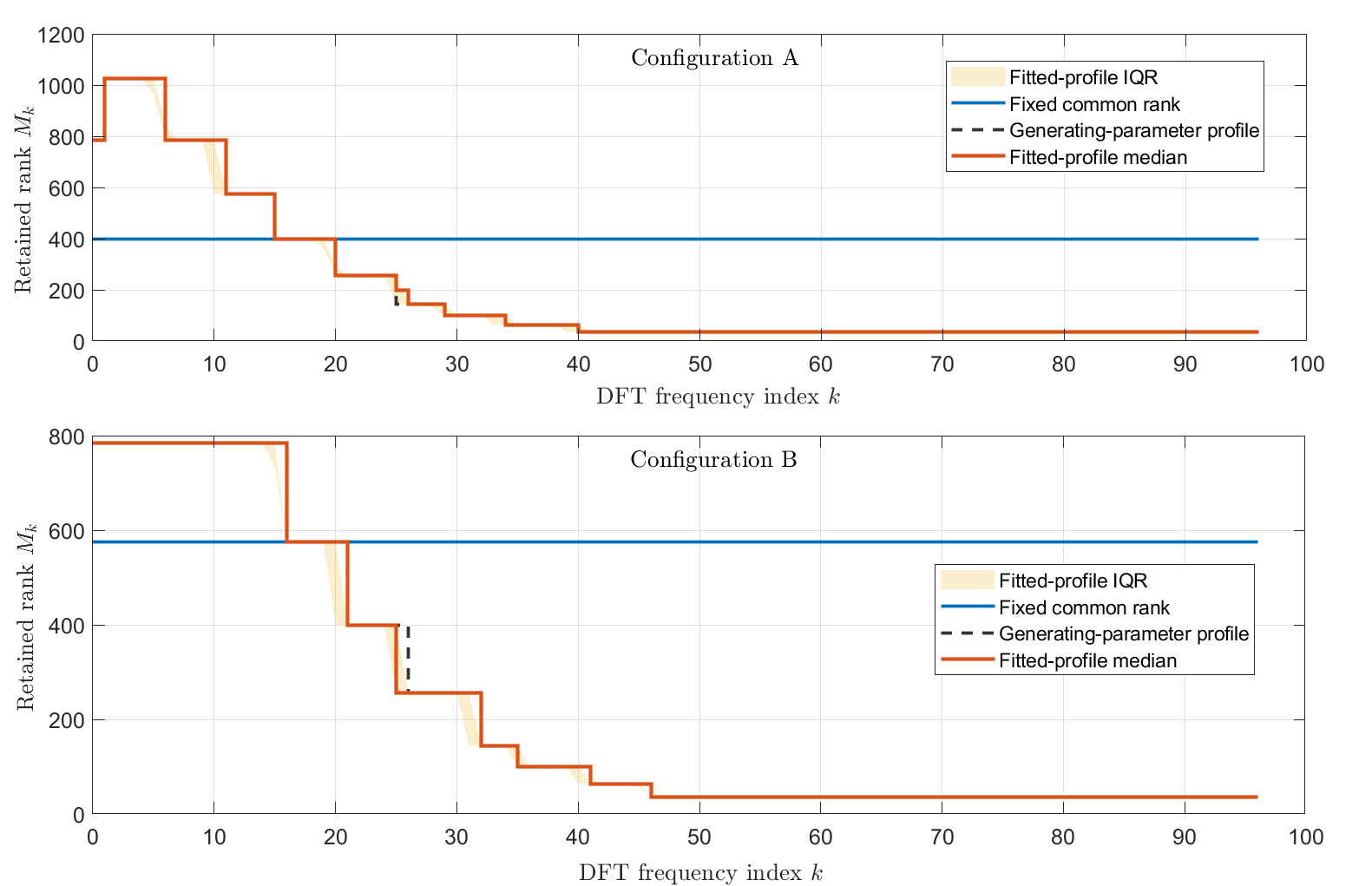}
\caption{Frequency-specific ranks over 20 Fuentes replicates.  The orange
curve and band show the median and interquartile range (25th--75th
percentiles) of ranks selected at the replicate-specific analytic-reference
estimates.  The dashed curve is
the profile obtained at the generating parameter, and the blue line is the
fixed common rank used in Table~\ref{tab:fuentes-prediction-rmse}.}
\label{fig:adaptive-fuentes-ranks}
\end{figure}

Every selected covariance met the prescribed KL threshold.  Figure~ \ref{fig:adaptive-fuentes-ranks} shows that the required rank varies across temporal frequencies and that this pattern is stable across the 20 fitted replicates.  The fitted-parameter profiles also closely follow the profile at the generating parameter.  As Block~II of Table~\ref{tab:fuentes-prediction-rmse} shows, allocating ranks by frequency reduced the reconstruction gap between the fixed-rank HSGP and the analytic full-rank reference in both configurations. Across frequencies and replicates, the selected ranks ranged from 36 to 1024 in Configuration A and from 36 to 784 in Configuration B.  They therefore exceeded the corresponding fixed ranks of 400 and 576 at some frequencies and were smaller at others.  The improvement in Block~II reflects a redistribution of basis functions across temporal frequencies, including additional modes where the fixed rank was insufficient.  The adaptive construction therefore redistributes modes across frequencies; it does not impose a uniformly lower rank.  The frequency-adaptive approximation consequently retains a small systematic error.  The next benchmark evaluates the complete internal-reference workflow and records its computational time.

\subsubsection{Numerical-inversion study: directly specified multiplier}
\label{sec:direct-spectral-experiment}

The second benchmark reverses the computational balance of the Fuentes study.  Its multiplier $S_\mathrm{d}$ is available in closed form, whereas its half-spectrum $H_\mathrm{d}$ requires numerical spatial inversion.  The model combines the standard discrete-time AR(1) spectral factor \citep{brockwelldavis1991} with a nonnegative, integrable isotropic spatial multiplier. In two spatial dimensions, the data-generating multiplier was
\begin{equation}
 S_{\mathrm{d}}(\bm\omega,\lambda;\bm\theta)
 =\frac{v_{\mathrm{f}}}{\sqrt{2\pi}}g_\rho(\lambda)
 q_{\kappa(\lambda)}(\lVert\bm\omega\rVert),
 \qquad
 q_\kappa(r)=\frac{9\sqrt{3}}{\kappa^2}
 \Big(1+\big(\frac r\kappa\big)^3\Big)^{-2},
 \label{eq:rational-cubic-multiplier}
\end{equation}
where
\begin{equation}
 g_\rho(\lambda)=
 \frac{1-\rho^2}{1+\rho^2-2\rho\cos\lambda},
 \qquad
 \kappa(\lambda)=\kappa_0
 \big(1+\gamma\sin^2(\lambda/2)\big)^{\frac{1}{2}}.
 \label{eq:rational-cubic-components}
\end{equation}
Here $g_\rho$ is normalized to have unit average over $[-\pi,\pi]$ and $\rho$ controls temporal persistence.  The parameter $\kappa_0$ is the zero-frequency spatial inverse range, while $\gamma$ controls how that inverse range changes with temporal frequency and hence induces nonseparability.  The constant $9\sqrt{3}$ normalizes the spatial factor so that $(2\pi)^{-2}\int_{\R^2}q_\kappa(\lVert\bm\omega\rVert)\,\mathrm d\bm\omega=1$.

For spatial lag $\bm h$, the corresponding half-spectrum can be written as
\begin{equation}
 H_{\mathrm d}(\bm h,\lambda;\bm\theta)
 =\frac{v_{\mathrm f}}{\sqrt{2\pi}}g_\rho(\lambda)
 R\!\left(\kappa(\lambda)\lVert\bm h\rVert\right),
 \qquad
 R(a)=\frac{9\sqrt{3}}{2\pi}\int_0^\infty
 \frac{uJ_0(au)}{(1+u^3)^2}\,\mathrm du,
 \label{eq:constructed-radial-half-spectrum}
\end{equation}
where $J_0$ is the Bessel function of the first kind and $R(0)=1$.  Thus $H_{\mathrm d}(\bm0,\lambda;\bm\theta) =v_{\mathrm f}g_\rho(\lambda)/\sqrt{2\pi}$, and the temporal normalization gives marginal variance $v_{\mathrm f}$.  The generating parameter was $(v_{\mathrm{f}},\rho,\kappa_0,\gamma,v_{\mathrm{n}})=(1,0.85,5,2,0.05)$.

An independent numerical engine evaluated the radial function $R$ in \eqref{eq:constructed-radial-half-spectrum} by applying panelwise high-order Gauss--Legendre quadrature \citep{davisrabinowitz1984} to the Hankel integral, followed by interpolation.  It was used both to simulate the fields and to define the numerical-inversion reference likelihood. For each replicate, the Fourier coefficients were generated independently at the nonredundant frequencies from their model covariances, completed by conjugate symmetry, and transformed to the time domain by the inverse DFT. This circular construction makes frequency decoupling exact and leaves the spatial covariance approximation as the principal distinction between the methods.  Twenty paired realizations used a $15\times15$ spatial grid, $T=128$, 180 fitting sites, 45 whole-site holdouts, and four deterministic optimizer starts.  The numerical-inversion and HSGP models were fitted independently to the same realization and spatial split.

An initial HSGP fit supplied the covariance-parameter estimate used for rank selection.  Algorithm~\ref{alg:adaptive-hsgp} then selected frequency-specific ranks from $\mathcal C=\{25,49,81,121,169,225,324\}$ using the internal reference ladder $\{400,625,900\}$.  The pilot rank was $M_{\mathrm{init}}=121$, the domain extension factor was 2, and the numerical inversion used 24-point Gauss-Legendre panels.  We set $\epsilon_{\mathrm{KL}}=0.01$ per fitting site and allowed at most three correction refits.  The numerical-inversion covariance was used for simulation, independent model fitting, and external comparison, and was not used to select the frequency-adaptive HSGP ranks.

The complete frequency-adaptive procedure was applied to the same 20 records and fitting-holdout splits as the numerical-inversion comparison.  The internal reference stabilized at $M_{\mathrm{ref}}=400$ for every replicate, all selected ranks satisfied the KL tolerance, and the resulting ranks ranged from 49 to 225.  The fixed high-rank comparator used $M_{\mathrm{ref}}=400$ at every frequency.  Two replicates retained their first selected rank profiles, while the other 18 required one correction refit.  All 20 final profiles passed the post-fit KL audit.
% , and the stored optimizer convergence indicator was positive for every final frequency-adaptive fit.

Table~\ref{tab:spectral-multiplier-results} compares reconstruction accuracy and the complete observed wall time.  Latent-field RMSE is evaluated against the simulated process $Z$, whereas observed-field RMSE is evaluated against $Y=Z+\mathcal E$ and therefore involves measurement noise.  Wall time includes model fitting and held-site reconstruction for all methods.  For the frequency-adaptive HSGP it also contains reference stabilization and rank selection.  Data simulation and construction of the shared quadrature table and nested basis bank are excluded from all three totals.

\begin{table}[H]
\centering
\small
\caption{Held-site reconstruction and complete comparison time for the
directly specified multiplier over 20 paired replicates.  Entries are the
replicate mean $\pm$ a 95\% paired-bootstrap half-width, obtained by
resampling the matched realizations.  The fixed high-rank HSGP uses
$M_{\mathrm{ref}}=400$ at every frequency.}
\label{tab:spectral-multiplier-results}
\begin{tabular}{@{}lccc@{}}
\toprule
Method & Latent RMSE & Observed RMSE & Wall time (s)\\
\midrule
Numerical-inversion reference
& $0.11275\pm0.00080$ & $0.25075\pm0.00098$ & $149.27\pm10.95$\\
Fixed high-rank HSGP
& $0.11276\pm0.00080$ & $0.25076\pm0.00099$ & $229.95\pm19.44$\\
Frequency-adaptive HSGP
& $0.11360\pm0.00075$ & $0.25088\pm0.00101$ & $76.59\pm4.84$\\
\bottomrule
\end{tabular}
\end{table}

The fixed high-rank HSGP reproduced the numerical-inversion reconstruction with their mean latent RMSEs differing by less than $2\times10^{-5}$.  The frequency-adaptive HSGP had a slightly larger latent RMSE, increasing the paired mean by $0.0009$, or about $1\%$, relative to numerical-inversion reference, and the marginal intervals overlap.
% but the paired comparison across identical realizations shows that this small difference is systematic.
% The fixed comparator deliberately uses the stabilized internal-reference rank
% at every frequency.  The table therefore measures the cost of carrying that
% reference rank through the likelihood; it does not claim superiority over
% every independently tuned common-rank HSGP.

The frequency-adaptive workflow required $76.59$ seconds on average, compared with $149.27$ seconds for numerical inversion and $229.95$ seconds for fixed high-rank HSGP.  The mean of the replicate-wise paired time reductions was $49\%$ relative to numerical inversion and $67\%$ relative to using $M_{\mathrm{ref}}=400$ at every frequency.  The saving arises because the high-rank covariances are used only during reference stabilization and rank audits, while the repeated likelihood and reconstruction calculations use the smaller frequency-specific ranks.
% These percentages are specific to the tested model, implementation, and hardware.

\subsubsection{Cross-experiment comparison with spatial baselines}
\label{sec:cross-dataset-baselines}

% Table~\ref{tab:cross-dataset-baselines} provides a compact comparison with five spatial comparator methods.  These are nearest-neighbour and Shepard inverse-distance interpolation \citep{shepard1968}, adaptive-resampling RFF \citep{huangetal2025rff}, Mat\'ern-$3/2$ RBF interpolation \citep{wendland2005}, and universal kriging \citep[Ch.~3]{cressie1993}.  The table combines Configuration~A of the Fuentes experiment, the direct-multiplier simulation, and the ERA5 and Pacific applications described in Sections~\ref{sec:era5-application} and \ref{sec:pacific-wind}.  Within each experiment, every method uses the same fitting and held-site partitions, and tuning is confined to fitting data.   \ref{app:baseline-settings} summarizes the comparator specifications.

Table~\ref{tab:cross-dataset-baselines} provides a compact comparison with five spatial comparator methods.  These are nearest-neighbour and Shepard inverse-distance interpolation \citep{shepard1968}, adaptive-resampling random Fourier feature (RFF) model \citep{huangetal2025rff}, Mat\'ern-$3/2$ radial basis function (RBF) interpolation \citep{wendland2005}, and universal kriging \citep{cressie1993}.  The table combines Configuration~A of the Fuentes experiment, the direct-multiplier simulation, and the ERA5 and Pacific applications described in Sections~\ref{sec:era5-application} and \ref{sec:pacific-wind}.  Within each experiment, every method uses the same fitting and held-site partitions, and tuning is confined to fitting data.   \ref{app:baseline-settings} summarizes the comparator specifications.

\definecolor{bestcolor}{HTML}{4169E1}
\definecolor{secondcolor}{HTML}{CF4F33}
\newcommand{\bestcell}[1]{\mathbf{#1}}
\newcommand{\secondcell}[1]{\underline{#1}}

\begin{table}[H]
\centering
\small
\setlength{\tabcolsep}{4pt}
\caption{Held-site RMSE across four experiments.  Entries are means
$\mathbin{\pm}$ 95\% half-widths.  Simulation intervals use bootstrap
resampling of the 20 independent replicates, while ERA5 and Pacific intervals use a moving-block residual bootstrap to retain short-range temporal dependence.
Simulation entries are latent-field RMSEs, and the wind-data entries are
observed-field RMSEs.
The smallest and second-smallest column means are shown in boldface and
with underline, respectively.}
\label{tab:cross-dataset-baselines}
\begin{tabular}{@{}lcccc@{}}
\toprule
Method & Fuentes A & Direct multiplier & ERA5 & Pacific wind \\
\midrule
Nearest neighbour
& $0.2665\pm0.0018$
& $0.3215\pm0.0028$
& $0.3714\pm0.0572$
& $1.0998\pm0.1166$ \\
Inverse-distance weighting
& $0.1265\pm0.0007$
& $0.1604\pm0.0013$
& $0.1373\pm0.0200$
& $0.4821\pm0.0289$ \\
Adaptive-resampling RFF
& $0.2701\pm0.0040$
& $0.3006\pm0.0040$
& $0.0911\pm0.0098$
& $0.2110\pm0.0092$ \\
RBF interpolation
& $0.1012\pm0.0005$
& $0.1238\pm0.0012$
& $\bestcell{0.0712\pm0.0085}$
& $0.1268\pm0.0048$ \\
Universal kriging
& $0.2110\pm0.0010$
& $0.2200\pm0.0011$
& $0.0755\pm0.0095$
& $0.1336\pm0.0054$ \\
Fixed-rank HSGP
& $\secondcell{0.0784\pm0.0007}$
& $\bestcell{0.1128\pm0.0008}$
& $0.0759\pm0.0088$
& $\secondcell{0.1252\pm0.0047}$ \\
Frequency-adaptive HSGP
& $\bestcell{0.0773\pm0.0007}$
& $\secondcell{0.1140\pm0.0008}$
& $\secondcell{0.0714\pm0.0087}$
& $\bestcell{0.1213\pm0.0045}$ \\
\bottomrule
\end{tabular}
\end{table}

Table~\ref{tab:cross-dataset-baselines} compares the held-site
RMSEs across the four experiments. Fixed-rank and frequency-adaptive
HSGP achieve the two lowest mean RMSEs in both simulation
experiments and the Pacific wind application.
For the ERA5 reanalysis data, frame-wise RBF interpolation
achieves the lowest mean RMSE, followed closely by
frequency-adaptive HSGP, with a difference of only
$0.0002\,\mathrm{m\,s^{-1}}$ and substantially overlapping
confidence intervals.
Frequency-adaptive HSGP achieves lower mean RMSEs than its
fixed-rank counterpart in the Fuentes, ERA5, and Pacific
experiments, whereas fixed-rank HSGP performs slightly
better in the direct-multiplier experiment.
Overall, frequency-adaptive HSGP ranks among the two most
accurate methods in all four experiments, while frame-wise
RBF interpolation remains competitive for the observed
wind fields.

\subsection{North Sea reanalysis applications}
\label{sec:north-sea-reanalysis}

We study two reanalysis applications for the same meteorological variable and geographical domain.  ERA5 is ECMWF's global atmospheric reanalysis \citep{hersbach2020}.  CERRA is a higher-resolution regional reanalysis for Europe \citep{ridal2024}.  From each archive we extracted six-hourly 10-m eastward wind over the North Sea box $54^\circ$-$58^\circ$N and $2^\circ$-$6^\circ$E.  ERA5 provides the moderate-resolution experiment below.  The denser CERRA grid is used to examine the effects of HSGP rank and record length.  In both cases, complete time series are withheld at the reconstruction sites.  The experiments thus assess spatial interpolation over an observed period and do not consider future-time forecasting.

\subsubsection{ERA5: reconstruction as the spatial basis is enriched}
\label{sec:era5-application}

The ERA5 experiment contains $T=360$ observations over 90 days on a $17\times17$ grid ($n=289$).  The sites were divided into five disjoint, approximately equal groups.  For each fold, the complete 90-day series at 57 or 58 sites was withheld and the model was fitted to the remaining 231 or 232 sites.  After the five fits, every site had served exactly once as an unseen reconstruction location.  All methods used these same five folds.  Within each fold, the analytic full-rank model and each HSGP model minimized their own Whittle objective and reconstructed the held-out sites using the resulting parameter estimates.  The Fuentes smoothness was fixed at $\nu=2.5$, and the HSGP domain-extension factor was 4.  The fixed-rank HSGP comparisons retained $J=32$ and $J=48$ Laplacian modes in each spatial coordinate, corresponding to two-dimensional basis functions with $M=J^2=1024$ and $2304$, respectively. We also report a frequency-adaptive HSGP fit based on the relaxed tolerances $\epsilon_{\mathrm{ref}}=0.03$ and $\epsilon_{\mathrm{KL}}=0.05$.  Its candidate grid ranged from $J=16$ to $92$ in increments of four, and at most three correction refits were allowed.  The final fitted profiles used $J\in\{52,56,60\}$, and every final candidate covariance passed the KL tolerance.

\begin{table}[H]
\centering
\small
\caption{Reconstruction of complete held-out site records in five folds for
the North Sea ERA5 field.  RMSE is averaged over the 289 held-out sites and reported as the mean
$\mathbin{\pm}$ a 95\% site-bootstrap half-width.  % This interval resamples sites, whereas Table~\ref{tab:cross-dataset-baselines} uses a temporal moving-block residual bootstrap to compare methods under one interval construction.  
RMSE values are in $\mathrm{m\,s^{-1}}$.  NLPD is the
mean negative log predictive density.  Lower values are better for both
criteria.}
\label{tab:era5-cv}
\begin{tabular}{@{}lcc@{}}
\toprule
Method & RMSE $\mathbin{\pm}$ 95\% half-width & NLPD \\
\midrule
Analytic full-rank reference & $0.0715\pm0.0031$ & $-1.1595$ \\
Fixed-rank HSGP, $J=32$ & $0.0921\pm0.0020$ & $-0.9146$ \\
Fixed-rank HSGP, $J=48$ & $0.0759\pm0.0023$ & $-1.1247$ \\
Frequency-adaptive HSGP & $0.0714\pm0.0030$ & $-1.1768$ \\
\bottomrule
\end{tabular}
\end{table}

Table~\ref{tab:era5-cv} shows a clear improvement as the fixed spatial basis is enriched: increasing $J$ from 32 to 48 reduced mean RMSE from $0.0921$ to $0.0759\,\mathrm{m\,s^{-1}}$ and improved NLPD.  The richer fixed basis moved both criteria toward the analytic full-rank reference.  The frequency-adaptive HSGP fit reached RMSE $0.0714\,\mathrm{m\,s^{-1}}$, essentially the same as the analytic full-rank value and the frame-wise RBF value in Table~\ref{tab:cross-dataset-baselines}.
% The overlapping intervals do not distinguish these methods in predictive accuracy.  
% The frequency-adaptive ranks still
% exceeded the number of fitting sites, and the correction loop did not
% stabilize, so the ERA5 experiment supports predictive accuracy at richer rank
% but not reduced-rank computation.

Figure~\ref{fig:era5-reconstruction} compares the analytic and $J=48$ HSGP reconstructions at a corner, an interior location, and a boundary location. At all three sites, the two conditional means are nearly superimposed across the 90-day record, including during sharp peaks and reversals.  Their shared departures from the held-out series are larger than the visible differences between the two covariance calculations.  The agreement at both interior and boundary sites also gives no indication of a location-specific loss of accuracy in these examples.  The corresponding plug-in intervals almost coincide at the scale of the figure.

\begin{figure}[h]
\centering
\includegraphics[width=0.98\textwidth]{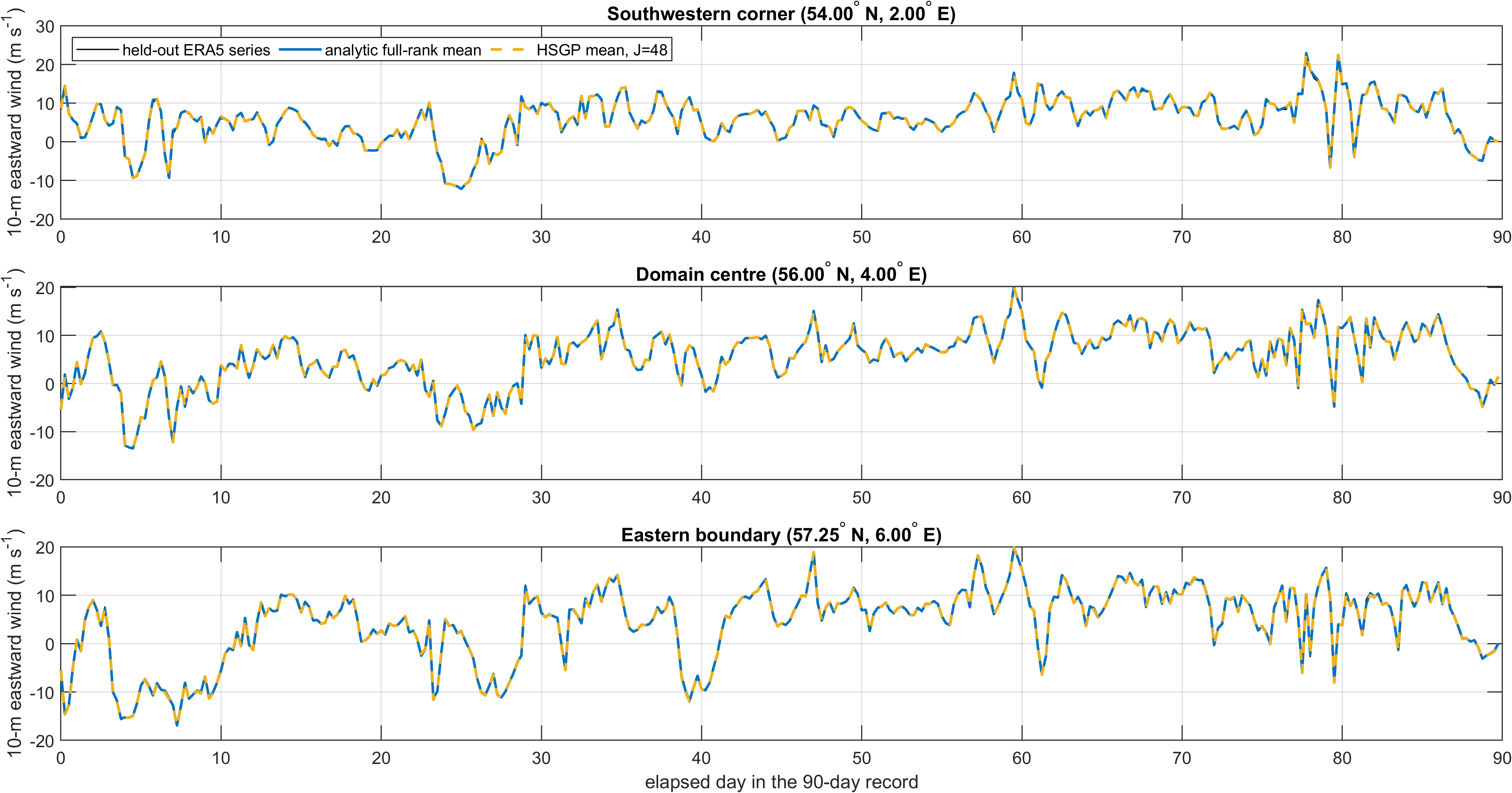}
\caption{Representative North Sea ERA5 reconstructions at a southwestern
corner site, a central site, and an eastern-boundary site.  The thin black line
is the held-out 10-m eastward wind, the solid blue line is the analytic
full-rank conditional mean, and the dashed gold line is the HSGP conditional
mean for $J=48$.  Light blue and gold ribbons are the corresponding point-wise
95\% plug-in observation intervals. The two ribbons are nearly
indistinguishable at this plotting scale.}
\label{fig:era5-reconstruction}
\end{figure}

\subsubsection{CERRA: basis rank and temporal record length}
\label{sec:cerra-application}

For the numerical experiment with CERRA dataset, we extracted 3701 ocean-grid locations in the same North Sea box.  A common set of 300 complete spatial series was reserved for evaluation, leaving 3401 fitting sites.  We analyzed a 31-day record with $T=124$ and a 90-day record with $T=360$, both sampled every six hours.  For each record, the HSGP used all 3401 training sites for model fitting, boundary factor $c_B=2$, and $M\in\{250,500,750,1000\}$ retained Laplacian basis functions.  The feature matrix and $\bm\Phi^\trans\bm\Phi$ were cached for each geometry.

\begin{table}[H]
\centering
\small
\caption{Pooled held-site RMSE ($\mathrm{m\,s^{-1}}$) for the full-grid CERRA
HSGP fits.  All rows use the same 3401 fitting and 300 held-out locations.  The
final column is the observed difference between the 90-day and 31-day
experiments.}
\label{tab:cerra-dense}
\begin{tabular}{@{}cccc@{}}
\toprule
$M$ & $T=124$ (31 days) & $T=360$ (90 days) & $\mathrm{RMSE}_{T360}-\mathrm{RMSE}_{T124}$ \\
\midrule
250  & 0.3013 & 0.2589 & $-0.0424$ \\
500  & 0.2139 & 0.1816 & $-0.0323$ \\
750  & 0.1673 & 0.1393 & $-0.0280$ \\
1000 & 0.1357 & 0.1112 & $-0.0245$ \\
\bottomrule
\end{tabular}
\end{table}

Table~\ref{tab:cerra-dense} shows two systematic patterns.  Held-site RMSE decreases at every increase in $M$, and the 90-day record has lower RMSE than the 31-day record at every tested rank.  The rank trend is consistent with a reduction in spatial basis-truncation error.  The record-length trend is also consistent with the finite-record theory in \eqref{eq:fejer-smoothed-spectrum}.  As $T$ increases, the Fej\'er window concentrates more tightly around $\lambda_k$, so the diagonal DFT covariance $\bm F^{(T)}_{kk}$ approaches the frequency-wise covariance $\bm F_k$ under the regularity conditions stated in Section~\ref{sec:temporal-dft}. 
% Only one window of each length was analyzed, however, so the RMSE difference does not isolate the effect of $T$ from other differences between the two records.

Figure~\ref{fig:cerra-spatial-snapshots} complements the pooled RMSE by showing the spatial distribution of error at the beginning, midpoint, and end of the 90-day record.  Spatial RMSE is $0.045$ and $0.048\,\mathrm{m\,s^{-1}}$ in the first and final panels, but rises to $0.100\,\mathrm{m\,s^{-1}}$ at the midpoint.  The larger midpoint error is concentrated mainly along the southern part of the domain, while the other two snapshots have smaller and more diffuse errors.  These snapshots show that a pooled RMSE can conceal substantial variation over both space and time.  The displayed surfaces interpolate errors at the 300 held-out sites solely for visualization.

\begin{figure}[H]
\centering
\includegraphics[width=0.99\textwidth]{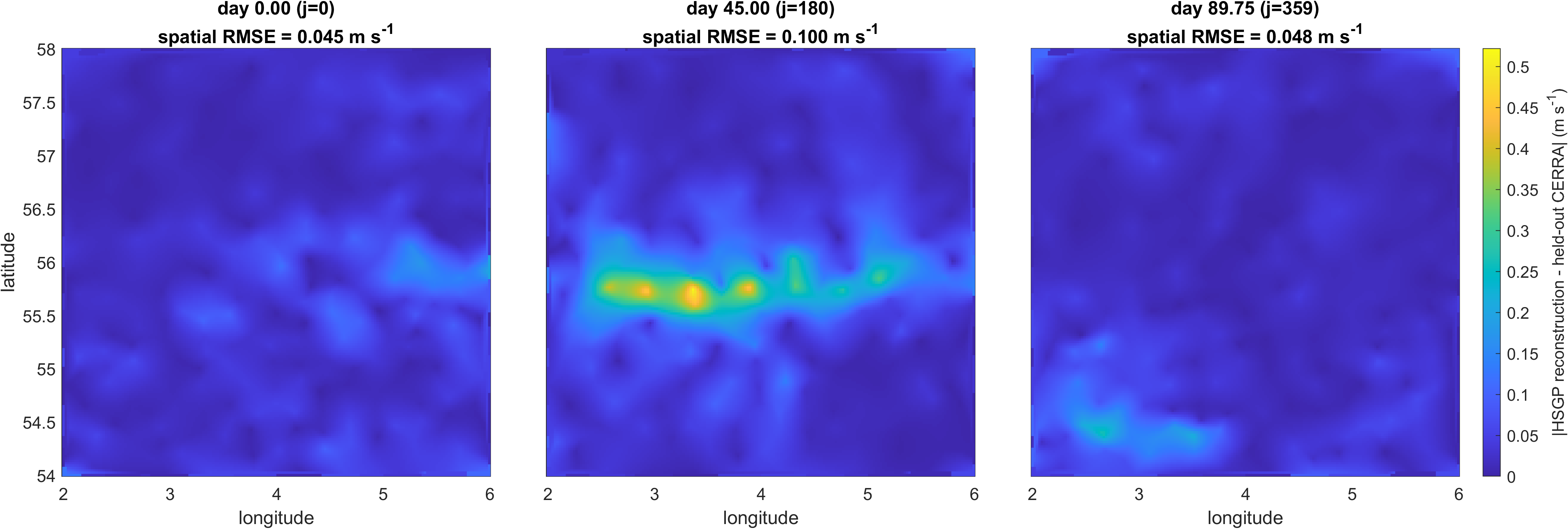}
\caption{Interpolated maps of absolute held-site reconstruction error for the
full-grid CERRA HSGP fit with $M=1000$ and $T=360$. Panels correspond to
$j=0$, 180, and 359, where $j=359=T-1$ is the final sampled index. Panel
titles report spatial RMSE across the original 300 held-out sites. 
%The panels share a common colour scale.
}
\label{fig:cerra-spatial-snapshots}
\end{figure}

\subsection{Measured wind-field applications}
\label{sec:measured-wind}

The final applications use measured wind records from two distinct spatial designs.  The tropical Pacific data form a regular oceanic grid, while the European data come from an irregular network of land stations.  In both applications, five-fold validation withholds the complete record at each selected site.  The Pacific study examines reconstruction, separability, and the spatial pattern of error.  The European study then compares the two procedures for selecting frequency-specific ranks introduced in Section~\ref{sec:adaptive-rank} on a larger station network.

\subsubsection{Tropical Pacific reconstruction and spatial error}
\label{sec:pacific-wind}
The first application revisits the tropical Pacific zonal-wind record analysed by \citet{cressiehuang1999}.  The data comprise $T=480$ six-hourly measurements from November 1992 through February 1993 on a regular $17\times17$ grid, giving $n=289$ spatial sites.  A balanced five-fold validation design withheld the complete time series at 57 or 58 sites per fold and fitted the model to the remaining 231 or 232 sites.

We compare fixed-rank HSGP fits under nonseparable and separable covariance models, together with the frequency-adaptive construction.  Under the separable restriction $\epsilon=1$, the half-spectrum factors into a frequency-dependent amplitude and a common spatial correlation function. Temporal frequency therefore changes the contribution to variance, but the shape and range of spatial correlation remain fixed.  Estimating $\epsilon$ allows the inverse range in \eqref{eq:fuentes-range}, and hence the spatial correlation, to vary with temporal frequency.  Held-site responses were used only for evaluation.  The mean model, basis sizes, rank ladders, tolerances, and stopping rules are collected in  \ref{app:measured-wind-settings}.

Table~\ref{tab:pacific-cv} reports two complementary criteria.  Mean absolute error measures the magnitude of held-site errors, and we also use negative log predictive density to evaluate the predictive variance.  Lower values indicate better performance for both measures.
\begin{table}[H]
\centering
\small
\caption{Complementary held-site criteria for the Pacific wind panel.  Mean
absolute error (MAE) is computed for each held-out time series and averaged
over the 289 sites.  Mean negative log predictive density (NLPD) evaluates
both the predictive mean and predictive variance.  %Lower values are better.
}
\label{tab:pacific-cv}
\begin{tabular}{@{}lcc@{}}
\toprule
Method & MAE ($\mathrm{m\,s^{-1}}$) & NLPD \\
\midrule
Fixed nonseparable HSGP & 0.1033 & 1.2105 \\
Fixed separable HSGP & 0.1035 & 1.2081 \\
Frequency-adaptive HSGP, $J=60$ & 0.0984 & $-0.2837$ \\
\bottomrule
\end{tabular}
\end{table}

% Table~\ref{tab:pacific-cv} shows almost identical predictive performance for the two fixed-rank models.  Their MAEs differ by only $0.0002\,\mathrm{m\,s^{-1}}$ and their NLPDs by 0.0024.  The fitted interaction parameter lies between 0.256 and 0.326 in every nonseparable fit, well below the separable value of one.  The estimated interaction therefore changes the covariance structure, but it provides no discernible improvement in these held-site criteria.

% The frequency-adaptive fit has the smallest MAE and a substantially smaller NLPD.  Every frequency selected the largest candidate rank, $J=60$, so the row measures the effect of using a uniformly richer basis; it does not show frequency-dependent rank reduction. 

In Table~\ref{tab:pacific-cv}, the two fixed-rank models show almost identical predictive performance using $M=2304$ spatial basis functions.  Their MAEs differ by only $0.0002\,\mathrm{m\,s^{-1}}$ and their NLPDs by $0.0024$.  The fitted interaction parameter lies between 0.256 and 0.326 in every nonseparable fit, well below the separable value of one.  The estimated interaction therefore changes the covariance structure, but it provides no discernible improvement in these held-site criteria. The frequency-adaptive fit has the smallest MAE and a substantially smaller NLPD.  Every frequency selected the largest candidate rank, $J=60$, so the row measures the effect of using a uniformly richer basis, and it does not show frequency-dependent rank reduction.

Figure~\ref{fig:pacific-reconstruction} compares the two reconstructions at three held-out sites on the central-latitude transect.  This fixed west-east selection shows the more difficult western boundary together with central and eastern locations. Both conditional means follow the observed trajectories closely in the interior and east, while the larger western discrepancies are consistent with the aggregate spatial pattern considered below.

\begin{figure}[h]
\centering
\includegraphics[width=0.79\textwidth]{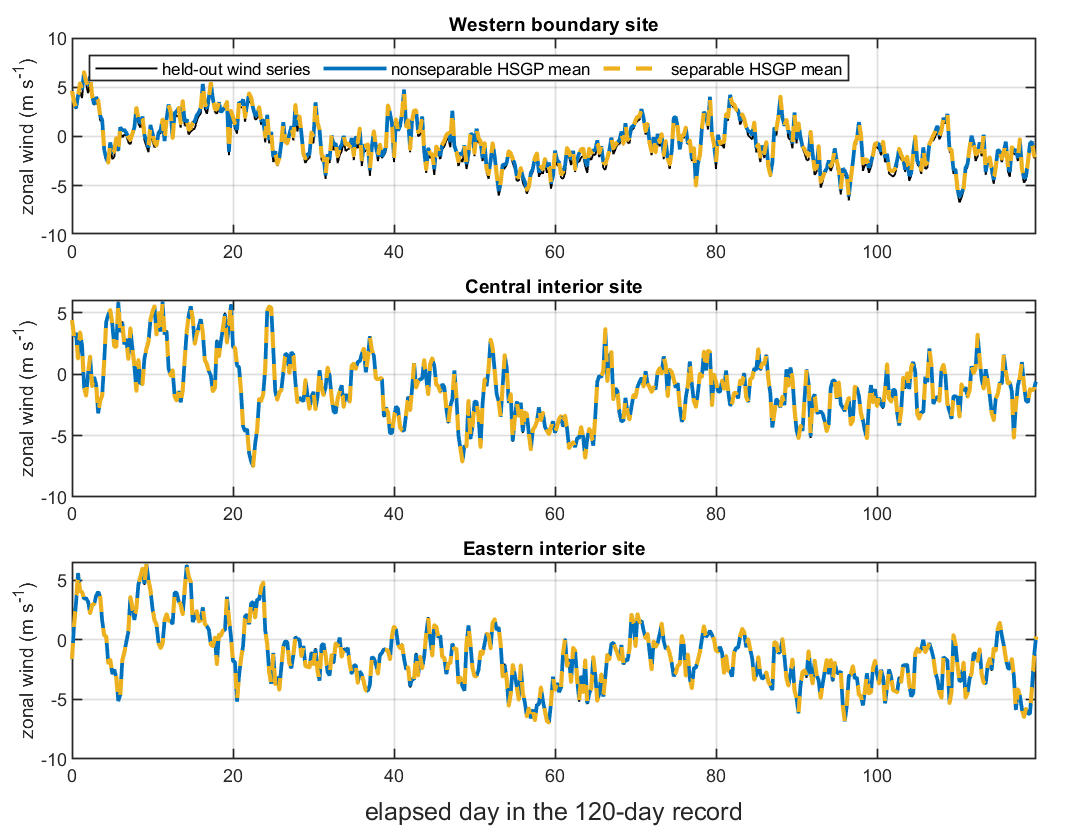}
\caption{Held-site zonal-wind reconstruction along the central-latitude
transect.  The thin black line is the observed series, the solid blue line is
the nonseparable HSGP conditional mean, and the dashed gold line is the
separable HSGP conditional mean.  Light ribbons show plug-in 95\% observation
intervals.  From top to bottom, the panels show a western boundary site, a
central interior site, and an eastern interior site on the same latitude.}
\label{fig:pacific-reconstruction}
\end{figure}

Figure~\ref{fig:pacific-spatial-rmse} reports the out-of-fold RMSE separately at every spatial site.  The two fitted covariance models have almost identical error surfaces, and their differences are small over most of the grid.  The main feature is a concentration of error near the western boundary that is common to both models.  Under the nonseparable fit, mean sitewise RMSE is $0.3854\,\mathrm{m\,s^{-1}}$ in the three westernmost columns.  It falls to $0.0937\,\mathrm{m\,s^{-1}}$ in the next two columns and to $0.0654\,\mathrm{m\,s^{-1}}$ across the remainder of the grid.  The spatial difficulty is therefore strongly localized instead of being distributed uniformly across the domain.

\begin{figure}[h]
\centering
\includegraphics[width=0.99\textwidth]{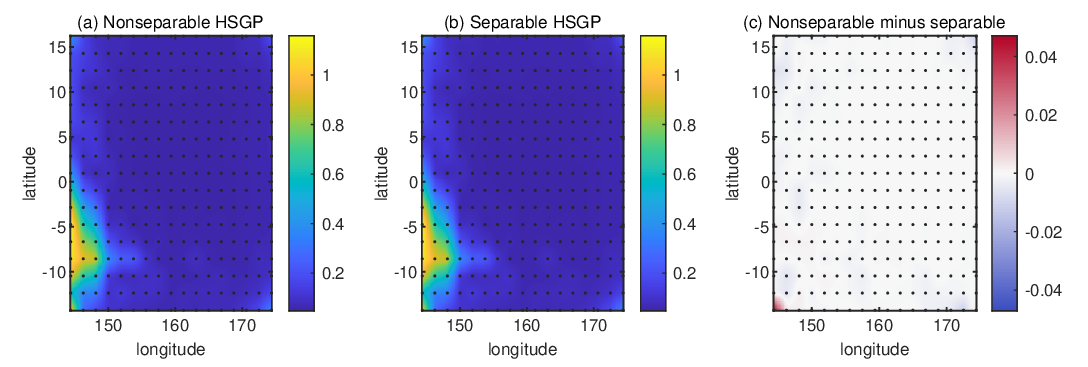}
\caption{Spatial distribution of cross-validation error from complete
held-out site records.  The
first two panels show naturally interpolated sitewise RMSE for the
nonseparable and separable HSGP fits on a common colour scale.  The third panel shows nonseparable minus separable RMSE in which positive values favour the separable restriction.  Black points mark the 289 evaluation sites, and interpolation is used to display the spatial pattern. Every numerical summary uses the original out-of-fold errors on each test site.}
\label{fig:pacific-spatial-rmse}
\end{figure}

The shared western concentration indicates spatial heterogeneity that remains unexplained by both stationary models.  Possible sources include boundary effects, spatially varying variance or nugget, and covariance nonstationarity. 

\subsubsection{European station data and rank-selection criteria}
\label{sec:ghcnh-wind}

The second application uses wind-speed observations from the NOAA Global Historical Climatology Network--hourly (GHCNh), a multisource collection from fixed land stations \citep{menne2023ghcnh}.  The analysis retained 839 European stations with complete observations at four times per day over 90 days, giving $T=360$.  The deterministic five-fold partition used 671 or 672 fitting stations and 167 or 168 held stations in each fold.

The experiment compares a fixed-rank HSGP with the two adaptive procedures of Section~\ref{sec:adaptive-rank}.  Algorithm~\ref{alg:adaptive-hsgp} evaluates the KL discrepancy directly, while Algorithm~\ref{alg:trace-adaptive-hsgp} uses its trace-based sufficient condition.  All three procedures used the same data, folds, mean structure, and covariance model.  Each procedure also estimated its own covariance parameters.  The fixed-rank fit retained $J=48$ at every temporal frequency.  Full preprocessing, basis, tolerance, and stopping specifications are given in  \ref{app:measured-wind-settings}.

\begin{table}[H]
\centering
\footnotesize
\setlength{\tabcolsep}{1.7pt}
\caption{Five-fold held-site results for the European GHCNh wind data.
Entries are pooled estimates $\pm$ the larger one-sided half-width of a 95\%
spatial-block bootstrap interval based on 5,000 resamples.  RMSE and MAE are
in $\mathrm{m\,s^{-1}}$, and lower values are better for all three criteria.
Audit time is the mean accumulated operational audit time per fold, including
initial rank selection and subsequent post-fit audits.
Wall time is the complete recorded time for one fold.
Detailed settings are reported in
 \ref{app:measured-wind-settings}.}
\label{tab:ghcnh-rank-rules}
\begin{tabular}{@{}lccccc@{}}
\toprule
HSGP procedure & RMSE & MAE
& NLPD & \shortstack{\ Mean audit time\\(min)\ }
& \shortstack{\ Wall time\\(min)} \\
\midrule
Fixed-rank HSGP, $J=48$
& $\ \bf{1.95\pm0.17}\ $ & $\ \bf{1.39\pm0.10}\ $ & $\ \bf{2.07\pm0.09}\ $ & -- & $31.0$ \\
Frequency-adaptive HSGP (Alg.~1)\ 
& $2.04\pm0.20$ & $1.46\pm0.11$ & $2.11\pm0.09$ & $6.72$ & $25.0$ \\
Frequency-adaptive HSGP (Alg.~2)\ 
& $1.99\pm0.19$ & $1.42\pm0.11$ & $2.09\pm0.09$ & $2.32$ & $21.3$ \\
\bottomrule
\end{tabular}
\end{table}

Table~\ref{tab:ghcnh-rank-rules} shows similar held-site errors for all three procedures.  The fixed-rank fit has the smallest point estimates, followed by Algorithm~\ref{alg:trace-adaptive-hsgp}, and all bootstrap intervals overlap. Because each procedure estimates its own covariance parameters, the table compares their complete fitting and reconstruction workflows.

% The main difference between the adaptive procedures is the cost of rank assessment.  Across the five folds, the direct KL rule made 25 audit calls and used 6.72 minutes per fold on average.  The trace rule made 18 calls and used 2.32 minutes.  An optional exact check of the trace-selected ranks added 1.15 minutes per fold.  Even with this check included, its mean audit time was 3.47 minutes, 48\% below that of the direct KL rule.  Table~ \ref{tab:ghcnh-rank-rules} reports the operational time of 2.32 minutes because the optional calculation verifies the selected ranks but does not alter them.

The main difference between the two adaptive procedures lies in the computational cost of rank assessment.
Across the five folds, the direct KL criterion (Algorithm~\ref{alg:adaptive-hsgp}) required 25 audit evaluations, averaging 6.72 minutes per fold, whereas
the trace criterion (Algorithm~\ref{alg:trace-adaptive-hsgp}) required 18 evaluations, averaging 2.32 minutes per fold. The total wall-clock times, including model fitting,
rank assessment, and reconstruction, were 31, 25, and 21 minutes for the fixed-rank method, the direct KL procedure, and the trace procedure, respectively.
The trace procedure terminated after two or three fits in each fold, with a maximum normalized KL discrepancy of 0.0045, below the prescribed tolerance
$\epsilon_{\mathrm{KL}}=0.01$.
Overall, the trace criterion achieved predictive errors comparable to those of the fixed-rank method, satisfied the exact KL tolerance, and incurred the lowest computational cost in this experiment.

\section{Discussion and conclusions}

This work combines a temporal DFT and Whittle likelihood with a spatial HSGP approximation for regularly sampled space-time data.  At each temporal frequency, the spatial covariance is constructed by evaluating the sampled spectral multiplier on a truncated Laplacian basis.  The proposed extension selects a separate spatial rank for each temporal frequency.  Algorithm~\ref{alg:adaptive-hsgp} compares candidate covariances with a stabilized high-rank HSGP reference using normalized Kullback-Leibler divergence. Algorithm~\ref{alg:trace-adaptive-hsgp} replaces that comparison with a sufficient covariance-scaled trace certificate.  Analytic or numerical-inversion covariances are used for validation when they are available.

The simulations distinguish approximation accuracy from computational performance.  In the experiment with Fuentes model, the analytic full-rank covariance verified the selected ranks against a known target.  In the directly specified multiplier experiment, the half-spectrum required numerical spatial inversion. Frequency adaptation increased mean latent-field RMSE by $1\%$ relative to reference values obtained by using high-order numerical quadrature, while reducing mean fitting-and-reconstruction time by $49\%$ relative to numerical inversion and by $67\%$ relative to retaining the reference rank at every frequency.  This experiment provides the clearest example of the intended regime: the multiplier is readily evaluated and the frequency-specific ranks remain well below the reference.

The real-data studies show that accurate HSGP reconstruction does not automatically imply substantial rank reduction. The Pacific fit selected the largest candidate rank at every frequency, while the CERRA results continued to improve as the fixed basis was enlarged. The European station experiment provides the clearest direct comparison of the two adaptive procedures. Under the implemented fit limits, both produced held-site errors comparable to those of the separately optimized fixed-rank model, with overlapping bootstrap intervals.
%Algorithm~\ref{alg:trace-adaptive-hsgp} selected ranks that passed the optional exact KL audit and reduced mean accumulated audit time from 6.72 to 2.32 minutes per fold relative to Algorithm~1. 
Numerical results further indicate that the trace criterion of Algorithm~\ref{alg:trace-adaptive-hsgp} can reduce the cost of rank assessment while preserving predictive accuracy, whereas Algorithm~\ref{alg:adaptive-hsgp} provides the more direct KL-based assessment. Frequency adaptation is most useful when the spectral multiplier can be evaluated directly, repeated calculations benefit from smaller frequency-specific bases, and these savings exceed the cost of selecting and auditing the ranks. Boundary effects, finite-record approximation, and covariance-parameter estimation remain separate sources of error that require further 

\section*{Data availability}

The empirical studies use publicly available ERA5 and CERRA reanalysis products, NCEP/NOAA Pacific wind data, and NOAA Global Historical Climatology Network hourly station observations. The MATLAB implementations, simulated-data generators, and code scripts used to reproduce the numerical experiments are available in a GitHub repository under an open-source license at: \url{https://github.com/XinHuang2022/DFT_HSGP}.

\section*{CRediT authorship contribution statement}

\noindent
\textbf{Xin Huang:} Writing -- review \& editing,
Writing -- original draft, Visualization, Software,
Methodology, Investigation, Validation, Formal analysis,
Conceptualization.

\noindent
\textbf{Jia Li:} Writing -- review \& editing,
Methodology, Formal analysis, Conceptualization, Resources.

\noindent
\textbf{Jun Yu:} Writing -- review \& editing,
Methodology, Formal analysis, Conceptualization, Resources.

\section*{Acknowledgements}
This work was supported by Kempe Stiftelserna project JCSMK23-0168. The computations used resources provided by the National Academic Infrastructure for Supercomputing in Sweden (NAISS) at the PDC Center for High Performance Computing, KTH Royal Institute of Technology, which is  partially funded by the Swedish Research Council through grant agreement no. 2022-06725. Jia Li's research is supported by the National Science Foundation under grant CCF-2205004.

\appendix
\section{Temporal sampling and finite-record Fourier calculations}
\label{app:temporal}

% \paragraph{Notation.}
% Throughout this manuscript, we use $n$ to denote the number of spatial sites, $T$ the number of regularly spaced observations at each site, and $\Delta t$ the sampling interval.  Indices $i,a,b$ refer to sites, $j$ to observation times, $k$ to DFT frequencies, $m$ to temporal aliases in Section~\ref{sec:temporal-dft} and to candidate ranks from Section~\ref{sec:adaptive-rank} onward, and $\bm j$ to a multi-index of spatial Laplacian modes.  Bold lower-case symbols are vectors and bold upper-case symbols are matrices.

\subsection{Fourier convention and temporal sampling}

For a function $g$ of $q$ Euclidean coordinates, the unitary Fourier pair is
\begin{equation}
 \widehat g(\bm\upsilon)=(2\pi)^{-q/2}
 \int e^{-\ii\bm\upsilon^\trans\bm x}g(\bm x)\,\mathrm d\bm x,
 \qquad
 g(\bm x)=(2\pi)^{-q/2}
 \int e^{\ii\bm\upsilon^\trans\bm x}\widehat g(\bm\upsilon)\,\mathrm d\bm\upsilon.
 \label{eq:unitary-fourier-pair}
\end{equation}
For a stationary spatial kernel $k$, the covariance-spectrum normalization is
\[
 k(\bm x)=(2\pi)^{-q}\int e^{\ii\bm\upsilon^\trans\bm x}
 S(\bm\upsilon)\,\mathrm d\bm\upsilon,
 \qquad S=(2\pi)^{q/2}\widehat k.
\]
Thus $S_{\mathrm c}=(2\pi)^{d/2}f$ and $S_{\mathrm d}=(2\pi)^{d/2}f_{\mathrm d}$ under the conventions of the main text.

Section~\ref{sec:temporal-dft} derives the temporal aliasing identity by partitioning the continuous frequency axis and comparing the resulting expression with the discrete-time Fourier-series inversion.  Sufficient integrability is required to interchange the sum and integral in \eqref{eq:main-folded-inverse} and to take the spatial transform term by term in \eqref{eq:sampled-joint-density}.  These conditions yield the sampled multiplier in \eqref{eq:sampled-spectrum-definition}.  A model specified directly through $S_{\mathrm d}$ requires neither a continuous-time precursor nor an alias sum.

\subsection{Whittle limit and circular model}

Equations~\eqref{eq:exact-dft-window-covariance} and \eqref{eq:fejer-smoothed-spectrum} give the exact finite-record covariance and its Fej\'er-window form.  With fixed $n$, weighted absolute summability of the temporal covariance sequence and a spectrum bounded away from singularity imply the standard Fourier-domain limit \citep[Ch.~10]{brockwelldavis1991}
\begin{equation}
 \sup_k\|\bm F^{(T)}_{kk}-\bm F(\lambda_k)\|\longrightarrow0,
 \qquad
 \sup_{k\ne k'}\|\bm F^{(T)}_{kk'}\|\longrightarrow0.
 \label{eq:whittle-diagonal-limit}
\end{equation}
This statement supports frequency decoupling for fixed collections of coefficients.  It is not an estimator-accuracy result for a likelihood whose dimension grows with $T$ or $n$.

Exact frequency decoupling instead follows from the explicitly circular covariance
\begin{equation}
 \bm\Gamma_Y^{\mathrm{circ}}(\ell)
 =\frac1T\sum_{k=0}^{T-1}\bm F(\lambda_k)e^{\ii\lambda_k\ell},
 \qquad
 \mathbb E(\bm d_k\bm d_{k'}^\herm)
 =\bm F(\lambda_k)\mathbb I(k=k').
 \label{eq:periodic-dft-covariance}
\end{equation}
Interior nonredundant coefficients are then proper complex Gaussian and independent, while zero and Nyquist coefficients are real Gaussian and negative frequencies are determined by conjugate symmetry.  The circular covariance is a distinct model, rather than the covariance of an ordinary finite record restricted to $T$ observations.

\section{Spatial operator and HSGP construction}
\label{app:spatial}

Propositions~\ref{prop:continuous-spatial-symbol} and \ref{prop:spatial-symbol} establish that the continuous and sampled spatial covariance operators are Fourier multipliers.  Section~\ref{sec:hsgp-covariance} then identifies a radial sampled multiplier with a scalar function of the whole-space Laplacian and transfers that function to the Dirichlet Laplacian on the extended rectangle $\Omega$.  Solving the coordinatewise eigenproblem gives
\begin{align}
 \phi_{\bm j}(\bm s)
 &=\prod_{r=1}^d L_r^{-1/2}
 \sin\!\left[\frac{\pi j_r}{2L_r}(s_r-c_r+L_r)\right],
 \label{eq:dirichlet-basis}\\
 \Lambda_{\bm j}
 &=\sum_{r=1}^d\left(\frac{\pi j_r}{2L_r}\right)^2,
 \qquad \rho_{\bm j}=\sqrt{\Lambda_{\bm j}}.
 \label{eq:laplacian-eigenvalue}
\end{align}
The basis is orthonormal with respect to integration over $\Omega$; its values at observation sites need not form an orthogonal matrix.  Applying $\zeta_{\mathrm d}$ to these eigenvalues and retaining $\bm j\in\mathcal J$ produces the finite kernel in \eqref{eq:hsgp-general-expansion}.

This bounded-domain spectral operator is generally different from the target integral operator restricted to $\Omega\times\Omega$.  Consequently, the Laplacian modes are not generally the Mercer eigenfunctions of that restricted kernel, and increasing $M$ at fixed $\Omega$ removes neither the imposed boundary condition nor its associated error.  Approximation to the whole-space kernel requires both a receding boundary and an increasing frequency cutoff, together with suitable regularity and spectral-tail conditions \citep[Sec.~4]{solinsarkka2020}.  For a tensor truncation, $M=\prod_rJ_r$ and $\omega_{\max,r}=\pi J_r/(2L_r)$. Therefore one sufficient scaling regime has $L_r\to\infty$ and $J_r/L_r\to\infty$ in every coordinate.

\section{Likelihood, computation, and reconstruction details}
\label{app:inference}

The main text gives the covariance likelihood and reconstruction formulas for a known zero mean.  For a fitted mean, let $\mathbb E(\bm y_j)=\bm X_j\bm\xi$, $\bm X_k^D=T^{-1/2}\sum_j\bm X_je^{-\ii\lambda_kj}$, and $\bm r_k(\bm\xi)=\bm d_k-\bm X_k^D\bm\xi$.  Replacing $\bm d_k$ by $\bm r_k$ in the likelihood gives the real information matrix and score
\begin{align}
 \bm{\mathcal I}_\xi
 &=\sum_{k\in\mathcal K_R}(\bm X_k^D)^\trans\bm F_k^{-1}\bm X_k^D
 +2\Re\sum_{k\in\mathcal K_+}(\bm X_k^D)^\herm\bm F_k^{-1}\bm X_k^D,
 \notag\\
 \bm g_\xi
 &=\sum_{k\in\mathcal K_R}(\bm X_k^D)^\trans\bm F_k^{-1}\bm d_k
 +2\Re\sum_{k\in\mathcal K_+}(\bm X_k^D)^\herm\bm F_k^{-1}\bm d_k.
 \label{eq:mean-information}
\end{align}
Thus $\widehat{\bm\xi}=\bm{\mathcal I}_\xi^{-1}\bm g_\xi$ and, up to parameter-independent constants,
\begin{equation}
 \mathcal L_{\mathrm{prof}}
 =\mathcal L_0-\tfrac12\bm g_\xi^\trans
 \bm{\mathcal I}_\xi^{-1}\bm g_\xi,
 \qquad
 \mathcal L_{\mathrm{REML}}
 =\mathcal L_{\mathrm{prof}}+\tfrac12\log|\bm{\mathcal I}_\xi|.
 \label{eq:whittle-reml}
\end{equation}
The analytic full-rank and HSGP fits apply the same mean treatment but use their respective covariance matrices and fitted parameters.

For HSGP, write $\bm G=\bm\Phi^\trans\bm\Phi$, $\bm b_k=\bm\Phi^\trans\bm d_k$, and $u_k=\bm d_k^\herm\bm d_k$ as in \eqref{eq:cached-gram}.  A formulation that also permits zero spectral weights uses
\begin{equation}
 \bm C_k=\bm I_M+v_{\mathrm n}^{-1}
 \bm W_k^{1/2}\bm G\bm W_k^{1/2},
 \qquad \bm z_k=\bm W_k^{1/2}\bm b_k.
 \label{eq:sqrt-system}
\end{equation}
The determinant lemma and Woodbury identity then give
\begin{equation}
 \log|\widetilde{\bm F}_k|
 =n\log v_{\mathrm n}+\log|\bm C_k|,
 \qquad
 \bm d_k^\herm\widetilde{\bm F}_k^{-1}\bm d_k
 =v_{\mathrm n}^{-1}u_k-v_{\mathrm n}^{-2}
 \bm z_k^\herm\bm C_k^{-1}\bm z_k.
 \label{eq:cached-feature-formulas}
\end{equation}
Mean-design projections can be cached in the same manner.  Observation-space factorization may be preferable when $M$ is not smaller than $n$.

Let $\bm V_{\ast,k}^{Z}$ and $\bm V_{\ast,k}^{Y}$ denote the latent-field and noisy-observation conditional covariance matrices at the reconstruction sites.  For disjoint fitting and reconstruction sites with independent measurement noise,
\begin{equation}
 \bm V_{\ast,k}^{Y}=\bm V_{\ast,k}^{Z}+v_{\mathrm n}\bm I_{n_\ast}.
 \label{eq:latent-observation-variance}
\end{equation}
After restoring conjugate frequencies, the working conditional covariance between time indices $j$ and $j'$ is, for $a\in\{Z,Y\}$,
\begin{equation}
 \operatorname{Cov}(\bm Y_{\ast,j}^{a},\bm Y_{\ast,j'}^{a}\mid\bm y,
 \widehat{\bm\theta},\widehat{\bm\xi})
 =\frac1T\Big(
 \sum_{k\in\mathcal K_R}e^{\ii\lambda_k(j-j')}\bm V_{\ast,k}^{a}
 +2\Re\sum_{k\in\mathcal K_+}e^{\ii\lambda_k(j-j')}\bm V_{\ast,k}^{a}
 \Big).
 \label{eq:time-variance}
\end{equation}
For an ordinary finite record these are Whittle conditional covariances. They condition on the fitted covariance and mean parameters and do not include covariance-parameter uncertainty.

\subsection{KL evaluation and internal-reference stability}
\label{app:adaptive-kl-details}

The discrepancy in \eqref{eq:adaptive-gaussian-kl} is evaluated without an explicit inverse.  For Cholesky factors $\bm P=\bm L_P\bm L_P^\herm$ and $\bm Q=\bm L_Q\bm L_Q^\herm$,
\begin{equation}
 \log|\bm P|=2\sum_{i=1}^n\log[\bm L_P]_{ii},
 \qquad
 \operatorname{tr}(\bm Q^{-1}\bm P)
 =\|\bm L_Q^{-1}\bm L_P\|_F^2.
 \label{eq:adaptive-kl-cholesky}
\end{equation}
The same expression with $\bm P$ and $\bm Q$ interchanged gives the reverse directed KL divergence.  A materially negative computed value signals a numerical error; only values at floating-point roundoff scale are truncated to zero.

For an internal reference, let $\mathcal R=\{r_1<\cdots<r_Q\}$ be a nested reference ladder with $r_1\geq m_L$.  At a fixed parameter vector, define
% \begin{equation}
%  \delta_q(\bm\theta)=
%  \max_{k\in\mathcal K}\frac{1}{n}
%  \max\left[
%  \mathcal D_k\{\bm Q_{k,r_q}(\bm\theta)\Vert
%                \bm Q_{k,r_{q+1}}(\bm\theta)\},
%  \mathcal D_k\{\bm Q_{k,r_{q+1}}(\bm\theta)\Vert
%                \bm Q_{k,r_q}(\bm\theta)\}
%  \right].
%  \label{eq:adaptive-reference-stability}
% \end{equation}
\begin{equation}
 \delta_q(\bm\theta)
 =
 \frac{1}{n}
 \max_{k\in\mathcal K}
 \max\left\{
 \mathcal D_k\!\left(
   \bm Q_{k,r_q}(\bm\theta)
   \,\Vert\,
   \bm Q_{k,r_{q+1}}(\bm\theta)
 \right),
 \mathcal D_k\!\left(
   \bm Q_{k,r_{q+1}}(\bm\theta)
   \,\Vert\,
   \bm Q_{k,r_q}(\bm\theta)
 \right)
 \right\}.
 \label{eq:adaptive-reference-stability}
\end{equation}
The reference rank is the first $r_q$ for which $\delta_q(\bm\theta)\leq\epsilon_{\mathrm{ref}}$; if no adjacent pair passes, the reference is unresolved and the ladder must be extended.  The two directions are checked because KL divergence is asymmetric.  When an internal reference is used, this stability check is repeated at each parameter vector entering the rank rule in \eqref{eq:adaptive-rank-rule}.  Requiring $M_{\mathrm{ref}}\geq m_L$ ensures that every candidate in $\mathcal C$ is contained in the internal reference.

If the reference is unresolved, the reference ladder is extended before rank selection continues.  If no eligible candidate satisfies the selected criterion, the candidate grid or extended domain must be enlarged.  During the fit-and-check cycle, accepted ranks are allowed to increase but not decrease.  Reaching the update limit while any frequency still requests a larger rank is reported as an unresolved rank profile.

\paragraph{Derivation and evaluation of the trace certificate.}
For the nested internal reference, $\bm P_k=\bm Q_{k,m}+\bm R_{k,m}$, with $\bm R_{k,m}$ defined in \eqref{eq:trace-tail-covariance}.  Standardizing the omitted covariance by the candidate covariance gives
\begin{equation}
 \bm A_{k,m}
 =\bm Q_{k,m}^{-1/2}\bm R_{k,m}\bm Q_{k,m}^{-1/2}\succeq\bm 0.
 \label{eq:trace-standardized-tail}
\end{equation}
Substitution into \eqref{eq:adaptive-normalized-kl} yields
\begin{equation}
 e_{k,m}(\bm\theta)
 =\frac{c_k}{n}\big(
   \operatorname{tr}(\bm A_{k,m})
   -\log\lvert\bm I_n+\bm A_{k,m}\rvert
  \big).
 \label{eq:trace-kl-identity}
\end{equation}
If $a_1,\ldots,a_n$ are the nonnegative eigenvalues of $\bm A_{k,m}$, then
\[
 \prod_{r=1}^n(1+a_r)
 \geq 1+\sum_{r=1}^n a_r
 =1+t_{k,m}.
\]
Applying this inequality to \eqref{eq:trace-kl-identity} gives the upper bound in \eqref{eq:trace-kl-bound}.

The trace is evaluated through linear solves rather than an explicit inverse. If $\bm Q_{k,m}=\bm L_{k,m}\bm L_{k,m}^{\trans}$, then
\begin{equation}
 t_{k,m}
 =\sum_{j=m+1}^{M_{\mathrm{ref}}}
 w_{j,k}\,
 \lVert\bm L_{k,m}^{-1}\bm\phi_j\rVert_2^2.
 \label{eq:solve-trace-evaluation}
\end{equation}
Thus candidate comparison requires a factorization of $\bm Q_{k,m}$ and triangular solves for the omitted basis columns, but not a factorization of $\bm P_k$.  The certificate relies on the positive-semidefinite basis tail in \eqref{eq:trace-tail-covariance}; it is not a bound for an arbitrary external target covariance.

For the Whittle product, the frequency-wise divergences add.  The present rule controls each normalized term separately and does not impose a global divergence budget; Section~\ref{sec:adaptive-rank} states the resulting scope of the guarantee.

\section{Model-specific formulas and supporting numerical results}

\subsection{Fuentes inversion and sampled multiplier}
\label{app:fuentes}

For the Fuentes spectrum, set $A(\tau)=\alpha^2(\beta^2+\tau^2)$, $B(\tau)=\beta^2+\epsilon\tau^2$, $\eta=\nu-d/2$, and $\kappa(\tau)=\big(A(\tau)/B(\tau)\big)^{1/2}$.  With $\Matern_\eta(x)=2^{1-\eta}x^\eta K_\eta(x)/\Gamma(\eta)$, the spatial inverse transform is
\begin{equation}
 \frac1{(2\pi)^{d/2}}\int_{\R^d}
 \frac{e^{\ii\bm\omega^\trans\bm h}}
 {(A+B\|\bm\omega\|^2)^\nu}\,\mathrm d\bm\omega
 =\frac{\Gamma(\nu-d/2)}{2^{d/2}\Gamma(\nu)}
 A^{d/2-\nu}B^{-d/2}
 \Matern_\eta(\kappa\|\bm h\|).
 \label{eq:joint-to-half-spectrum}
\end{equation}
After absorbing the common gamma factor into the positive amplitude normalizer, write
\[
 p(\tau)=A(\tau)^{d/2-\nu}B(\tau)^{-d/2},
 \qquad
 H_{c,\mathrm{raw}}(\bm h,\tau)=p(\tau) \Matern_\eta\big(\kappa(\tau)\|\bm h\|\big).
\]
For alias cutoff $Q$, let $\tau_m=(\lambda+2\pi m)/\Delta t$, $p_{m\lambda}=\Delta t^{-1}p(\tau_m)$, and $\kappa_{m\lambda}=\kappa(\tau_m)$.  The normalized sampled half-spectrum is
\begin{equation}
 H_{\mathrm d}^{(Q)}(\bm h,\lambda;\bm\theta)
 =v_{\mathrm f}a_{Q,\Delta t}(\bm\theta)
 \sum_{m=-Q}^{Q}p_{m\lambda}
 \Matern_\eta(\kappa_{m\lambda}\|\bm h\|),
 \qquad
 a_{Q,\Delta t}^{-1}
 =\frac1{\sqrt{2\pi}}\int_{-\pi}^{\pi}
 \sum_{m=-Q}^{Q}p_{m\lambda}\,\mathrm d\lambda.
 \label{eq:discrete-half-spectrum}
\end{equation}
This normalization makes $\int_{-\pi}^{\pi}H_{\mathrm d}^{(Q)}(\bm0,\lambda) \,\mathrm d\lambda/\sqrt{2\pi}=v_{\mathrm f}$.

Define the unit-integral Mat\'ern density
\begin{equation}
 \rho_\eta(\bm\omega;\kappa)
 =\frac{\Gamma(\eta+d/2)}{\Gamma(\eta)\pi^{d/2}}
 \frac{\kappa^{2\eta}}
 {(\|\bm\omega\|^2+\kappa^2)^{\eta+d/2}}.
 \label{eq:normalized-matern-density}
\end{equation}
The sampled spatial multiplier corresponding to \eqref{eq:discrete-half-spectrum} is
\begin{equation}
 S_{\mathrm d}^{(Q)}(\bm\omega,\lambda;\bm\theta)
 =v_{\mathrm f}a_{Q,\Delta t}(\bm\theta)(2\pi)^d
 \sum_{m=-Q}^{Q}p_{m\lambda}
 \rho_\eta(\bm\omega;\kappa_{m\lambda}).
 \label{eq:frequency-specific-spatial-spectrum}
\end{equation}
The analytic full-rank reference and HSGP therefore use the same normalized finite-alias model; they differ only in the spatial covariance calculation.

\subsection{Fuentes parameter}
\label{app:fuentes-supporting-diagnostics}

Here ``reference'' denotes the analytic full-rank Fuentes covariance calculation, not known parameters.  Table~\ref{tab:parameter-recovery} summarizes the fitted parameters.
%, while Table~\ref{tab:fuentes-covariance-errors} compares the resulting analytic full-rank latent covariance stacks.

\begin{table}[htbp]
\centering
\small
\caption{Median parameter estimates over the 20 paired Fuentes replicates.
These are the fixed-noise fits used in
Table~\ref{tab:fuentes-prediction-rmse}, with $v_{\mathrm n}=0.05$.}
\label{tab:parameter-recovery}
\begin{tabular}{@{}clrrr@{}}
\toprule
Configuration & Parameter source & $\alpha$ & $\beta$ & $\epsilon$ \\
\midrule
\multirow{3}{*}{A} & Generating value & 1.500 & 0.350 & 0.350 \\
 & Analytic-reference fit & 1.498 & 0.347 & 0.333 \\
 & HSGP fit & 1.721 & 0.368 & 0.464 \\
\addlinespace
\multirow{3}{*}{B} & Generating value & 5.000 & 0.350 & 0.350 \\
 & Analytic-reference fit & 5.038 & 0.354 & 0.372 \\
 & HSGP fit & 5.557 & 0.370 & 0.468 \\
\bottomrule
\end{tabular}
\end{table}

% The covariance error below excludes the nugget and measures displacement of the fitted analytic full-rank covariance stack from the generating stack.  It is therefore not a fixed-parameter HSGP truncation error.

% \begin{table}[htbp]
% \centering
% \small
% \caption{Mean relative Frobenius error between analytic full-rank latent
% frequency-covariance stacks for the fixed-noise Fuentes fits.  Uncertainty is
% reported as the larger one-sided half-width of the paired-bootstrap 95\%
% interval for the Monte Carlo mean.}
% \label{tab:fuentes-covariance-errors}
% \begin{tabular}{@{}clc@{}}
% \toprule
% Configuration & Covariance comparison & Mean error (\%) \\
% \midrule
% \multirow{3}{*}{A}
%  & Analytic-reference fit versus generating value & $2.74\pm0.86$\\
%  & HSGP fit versus generating value & $6.47\pm0.77$\\
%  & HSGP fit versus analytic-reference fit & $6.09\pm0.80$\\
% \addlinespace
% \multirow{3}{*}{B}
%  & Analytic-reference fit versus generating value & $2.16\pm0.49$\\
%  & HSGP fit versus generating value & $9.99\pm0.67$\\
%  & HSGP fit versus analytic-reference fit & $8.90\pm0.38$\\
% \bottomrule
% \end{tabular}
% \end{table}

\subsection{Measured-wind experiment specifications}
\label{app:measured-wind-settings}

The main text reports the data scale and design features needed to interpret the results.  Tables~\ref{tab:measured-wind-data-settings} and \ref{tab:measured-wind-rank-settings} collect the remaining implementation specifications.  Whole-location validation withholds the complete temporal record at each evaluation site.

\begin{table}[H]
\centering
\footnotesize
\caption{Data and validation settings for the measured-wind applications.}
\label{tab:measured-wind-data-settings}
\begin{tabularx}{\textwidth}{@{}p{1.65cm}XXX@{}}
\toprule
Application & Observations & Spatial validation & Mean model \\
\midrule
Tropical Pacific
& Zonal wind on a $17\times17$ grid ($n=289$); $T=480$ six-hourly
  observations from November 1992 to February 1993
& Balanced five-fold split; 231 or 232 fitting sites and 57 or 58 held sites
  per fold
& Intercept and linear east-west and north-south trends \\
European GHCNh
& Wind speed at 839 stations within $35^\circ$--$72^\circ$N and
  $12^\circ$W--$35^\circ$E; observations at 00, 06, 12, and 18 UTC for 90
  days from 1 June 2024 ($T=360$)
& Deterministic five-fold split; 671 or 672 fitting stations and 167 or 168
  held stations per fold
& Spatial plane, diurnal sine and cosine terms, and a semidiurnal cosine term \\
\bottomrule
\end{tabularx}
\end{table}

For the European record, each target time was matched to the nearest quality-controlled observation within 60 minutes.  Requiring complete records at the four target times avoided temporal imputation.

\begin{table}[H]
\centering
\footnotesize
\caption{HSGP and rank-control settings for the measured-wind applications.
The spatial rank is $M=J^2$ throughout.}
\label{tab:measured-wind-rank-settings}
\begin{tabularx}{\textwidth}{@{}p{2.35cm}XX@{}}
\toprule
Application and procedure & Covariance and basis & Rank selection or stopping rule \\
\midrule
Pacific, fixed rank
& Fuentes multiplier with $\nu=2.5$, $c_B=4$, and $J=48$.  The
  nonseparable fit estimates $\epsilon$; the separable fit fixes $\epsilon=1$.
& A common rank is used at every retained temporal frequency. \\
Pacific, Alg.~\ref{alg:adaptive-hsgp}
& The same multiplier and folds.  The reference ladder is
  $J\in\{64,72,80,88,96\}$; the master candidate grid is
  $J=16,20,\ldots,92$.
& The selected reference is $J=64$, leaving eligible candidates
  $J=16,20,\ldots,60$.  The tolerances are
  $(\epsilon_{\mathrm{ref}},\epsilon_{\mathrm{KL}})=(0.03,0.05)$, with at
  most three correction refits. \\
European, fixed rank
& Fuentes multiplier with $\nu=2.5$, boundary-extension factor 2.5, and
  $J=48$.
& Covariance parameters are estimated separately in every fold. \\
European, Alg.~\ref{alg:adaptive-hsgp}
& Candidate ladder
  $J\in\{8,10,12,15,18,20,24,28,32,36,40,45\}$, pilot $J=10$, and
  reference $J_{\mathrm{ref}}=48$.
& $\epsilon_{\mathrm{KL}}=0.01$; at most four adaptive fits; stop when no
  more than one of 181 frequencies requests enlargement. \\
European, Alg.~\ref{alg:trace-adaptive-hsgp}
& The same multiplier, candidate ladder, pilot, and reference.
& $\epsilon_{\mathrm{KL}}=0.01$; at most three adaptive fits; stop when no
  more than nine of 181 frequencies requests enlargement. \\
\bottomrule
\end{tabularx}
\end{table}

The Pacific tolerances were fixed before examining held-site responses and were used as a computational sensitivity setting rather than a confirmatory calibration.  The European experiment prespecified $J_{\mathrm{ref}}=48$ instead of repeating reference-ladder stabilization at every audit.  Its practical fit caps were examined in a fold-1 timing and stopping sensitivity analysis and then held fixed for the remaining folds.  Held-site responses were not used to choose them.

\subsection{Spatial comparator specifications}
\label{app:baseline-settings}

Table~\ref{tab:cross-dataset-baselines} uses the held-site partitions of the corresponding HSGP experiments.  Table~\ref{tab:spatial-comparator-settings} summarizes the comparator implementations.

\begin{table}[H]
\centering
\footnotesize
\caption{Implementation settings for the spatial comparators.}
\label{tab:spatial-comparator-settings}
\begin{tabularx}{\textwidth}{@{}p{2.6cm}XX@{}}
\toprule
Method & Spatial specification & Tuning or additional settings \\
\midrule
Nearest neighbour
& Value at the closest fitting site
& No fitted hyperparameter \\
Inverse-distance weighting
& Four closest fitting sites with inverse-squared-distance weights
& Fixed neighbourhood and distance exponent \\
Mat\'ern-$3/2$ RBF
& Frame-wise radial-basis interpolation with a planar mean
& Length scale and ridge parameter selected by four-fold validation within
  the fitting sites \\
Universal kriging
& Frame-wise Mat\'ern-$3/2$ covariance with a planar mean
& Range fixed at four median nearest-neighbour spacings \\
Adaptive-resampling RFF
& Amplitude-based resampling of \citet{huangetal2025rff}
& 512 frequencies, 100 resampling rounds, proposal scale 0.2, and ridge
  parameter 0.01 \\
Fixed-rank HSGP
& $M=400$ in the simulations and $J=48$ per coordinate in the wind fields
& Covariance parameters estimated within each fitting fold \\
Frequency-adaptive HSGP
& Analytic target for the Fuentes validation; internal HSGP reference for the
  direct-multiplier, ERA5, and Pacific experiments
& The Pacific selection ends at common $J=60$ with KL-tolerence 0.01. \\
\bottomrule
\end{tabularx}
\end{table}

\end{document}